\documentclass[ba,preprint]{imsart}
\pubyear{2024}
\volume{TBA}
\issue{TBA}
\firstpage{1}
\lastpage{1}
\RequirePackage{amsthm,amsmath,amssymb,mathrsfs,dsfont,mathtools}
\RequirePackage[authoryear]{natbib}%% uncomment this for author-year citations
\RequirePackage[colorlinks,citecolor=black,linkcolor=black,urlcolor=blue,backref=page,backref=page]{hyperref}
\RequirePackage{graphicx}

\startlocaldefs

\usepackage{bm}
\usepackage{graphicx}
\usepackage{url}
\usepackage{booktabs,subcaption}
\usepackage{algorithm}
\usepackage{algpseudocode}
\usepackage{xcolor}
\usepackage{adjustbox}
\usepackage{lineno} 

\algrenewcommand\algorithmicrequire{\textbf{Input:}}
\algrenewcommand\algorithmicensure{\textbf{Output:}}

\usepackage{url}
\usepackage[noabbrev,capitalise]{cleveref}
\creflabelformat{equation}{#2\textup{#1}#3}

\usepackage{chapterbib}

\newtheorem{theorem}{Theorem}

\theoremstyle{definition}

\newcommand{\E}{\mathbb{E}}

\DeclareMathOperator*{\argmax}{arg\,max}
\DeclareMathOperator*{\argmin}{arg\,min}

\definecolor{revision_color}{HTML}{00008B}
\usepackage[shortlabels]{enumitem}

\endlocaldefs

\begin{document}

%% *** Frontmatter *** 

\begin{frontmatter}
\title{Improved control of Dirichlet location and scale near the boundary}
%\title{\support{}}
\runtitle{Improved control of Dirichlet location and scale near the boundary}

% Possible titles:
% Improved control of Dirichlet location and scale near the boundary
% Improved near-boundary control for Dirichlet priors

\begin{aug}
%%%%%%%%%%%%%%%%%%%%%%%%%%%%%%%%%%%%%%%%%%%%%%%
%% Only one address is permitted per author. %%
%% Only division, organization and e-mail is %%
%% included in the address.                  %%
%% Additional information can be included in %%
\author[A]{\fnms{Catherine}~\snm{Xue}\ead[label=e1]{csxue@g.harvard.edu}\orcid{0000-0002-5518-122X}}
\author[B]{\fnms{Alessandro}~\snm{Zito}\ead[label=e2]{azito@hsph.harvard.edu}\orcid{0000-0001-8515-6503}},
\and
\author[C]{\fnms{Jeffrey W.}~\snm{Miller}\ead[label=e3]{jwmiller@hsph.harvard.edu}\orcid{0000-0001-7718-1581}}
%%%%%%%%%%%%%%%%%%%%%%%%%%%%%%%%%%%%%%%%%%%%%%
%% Addresses                                %%
%%%%%%%%%%%%%%%%%%%%%%%%%%%%%%%%%%%%%%%%%%%%%%
\address[A]{Department of Biostatistics,
Harvard University\printead[presep={,\ }]{e1}}
\address[B]{Department of Biostatistics,
Harvard University\printead[presep={,\ }]{e2}}
\address[C]{Department of Biostatistics,
Harvard University\printead[presep={,\ }]{e3}}

\runauthor{C. Xue, A. Zito, and J. W. Miller}
\end{aug}

\begin{abstract}
Dirichlet distributions are commonly used for modeling vectors in a probability simplex. When used as a prior or a proposal distribution, it is natural to set the mean of a Dirichlet to be equal to the location where one wants the distribution to be centered.  However, if the mean is near the boundary of the probability simplex, then a Dirichlet distribution becomes highly concentrated either (i) at the mean or (ii) extremely close to the boundary.
% , in which case the mean is so far in the tail that it is not a useful representation of the center of the distribution for most purposes.   
% very little mass is near the mean.
Consequently, centering at the mean provides poor control over the location and scale near the boundary.
In this article, we introduce a method for improved control over the location and scale of Beta and Dirichlet distributions. Specifically, given a target location point and a desired scale, we maximize the density at the target location point while constraining a specified measure of scale. 
We consider various choices of scale constraint, such as fixing the concentration parameter, the mean cosine error, or the variance in the Beta case.
% Alternatively, one can maximize the density subject to other constraints, such as having a given concentration parameter.
In several examples, we show that this maximum density method provides superior performance for constructing priors, defining Metropolis-Hastings proposals, and generating simulated probability vectors.
% We demonstrate on an application to Bayesian analysis of mutational signatures in cancer genomics.
\end{abstract}

%\begin{keyword}[class=MSC]
%\kwd[Primary ]{62F15}
%\kwd{62H30}
%\end{keyword}

%% ** Keywords **

\end{frontmatter}

% \begin{bibunit}
% \input{main.tex}
% \putbib
% \end{bibunit}

%%%%%%%%%%%%%%%%%%%%%%%%%%%%%%%%%%%%%%%%%%%%%
% Introduction
%%%%%%%%%%%%%%%%%%%%%%%%%%%%%%%%%%%%%%%%%%%%%
\section{Introduction}

The family of Dirichlet distributions is the default choice when modeling probability vectors, due to its flexibility and analytical tractability.
% is the cornerstone of probability distributions over the simplex, and it is ubiquitous in the statistical literature thanks to its flexibility, analytical tractability, and straightforward sampling algorithms. 
As a conjugate prior for the parameters of a multinomial distribution, the Dirichlet plays a central role in many Bayesian models such as latent Dirichlet allocation \citep{blei2003latent}, finite and infinite mixture models \citep{Rousseau_Mengersen_2011, Richardson_green_1997, Miller_Harrison_2018}, 
admixture models \citep{pritchard2000inference},
non-negative matrix factorization \citep{zito2024compressivebayesiannonnegativematrix}, and variable selection with global-local shrinkage \citep{Bhattacharya_2015, Zhang_Bondell_2018}.

% Owing to their convenient conjugacy property to the multinomial distribution, which leads to simple Markov Chain Monte Carlo estimation procedures, Dirichlet priors have played a central role in a variety of popular Bayesian models as well, most noticeably 

When specifying the parameters of a $\mathrm{Dirichlet}(\alpha u_1,\ldots,\alpha u_K)$ distribution, where $\sum_{i=1}^K u_i = 1$ and $\alpha > 0$, the standard approach is to set the mean $(u_1,\ldots,u_K)$ equal to the location where one would like the distribution to be centered and adjust the concentration parameter $\alpha$ to control the scale.
% When specifying the parameters of a $\mathrm{Dirichlet}(a_1,\ldots,a_K)$ distribution, the standard approach is to set the mean equal to the location where one would like the distribution to be centered and adjust the concentration parameter $\alpha = \sum_{i=1}^K a_i$ to control the scale. 
% \aletodo{Here, $\alpha$ has not been formally defined. I would say that $\alpha$ is the sum of the Dirichlet parameters.}
Unfortunately, if the mean is close to the boundary of the probability simplex---that is, if one or more entries is near zero---then this approach exhibits serious problems.
First, if the mean is near the boundary, then $\alpha$ provides little control over the scale of the distribution.  
Naturally, for large values of $\alpha$, the distribution is highly concentrated around the mean.
Intuitively, one might expect that using a smaller concentration  would lead to a less informative prior.
However, as $\alpha$ decreases, the distribution does not become much more spread out -- and, in fact, the bulk of the mass just moves even farther towards the boundary, becoming concentrated extremely close to it.
% , and a small amount of mass is spread across the simplex.
As a result, the mean ends up located so far in the tail that it is no longer a useful representation of the center of the distribution.
\cref{fig:problem} illustrates this behavior in the case of a Beta distribution, which is representative of the general Dirichlet case since the marginals of a Dirichlet are Beta-distributed.
Thus, the mean and concentration parameter are not useful for controlling the location and scale of Dirichlet distributions near the boundary.

\begin{figure}
    \centering
    \includegraphics[trim=0 0 0 0, clip, width=\textwidth]{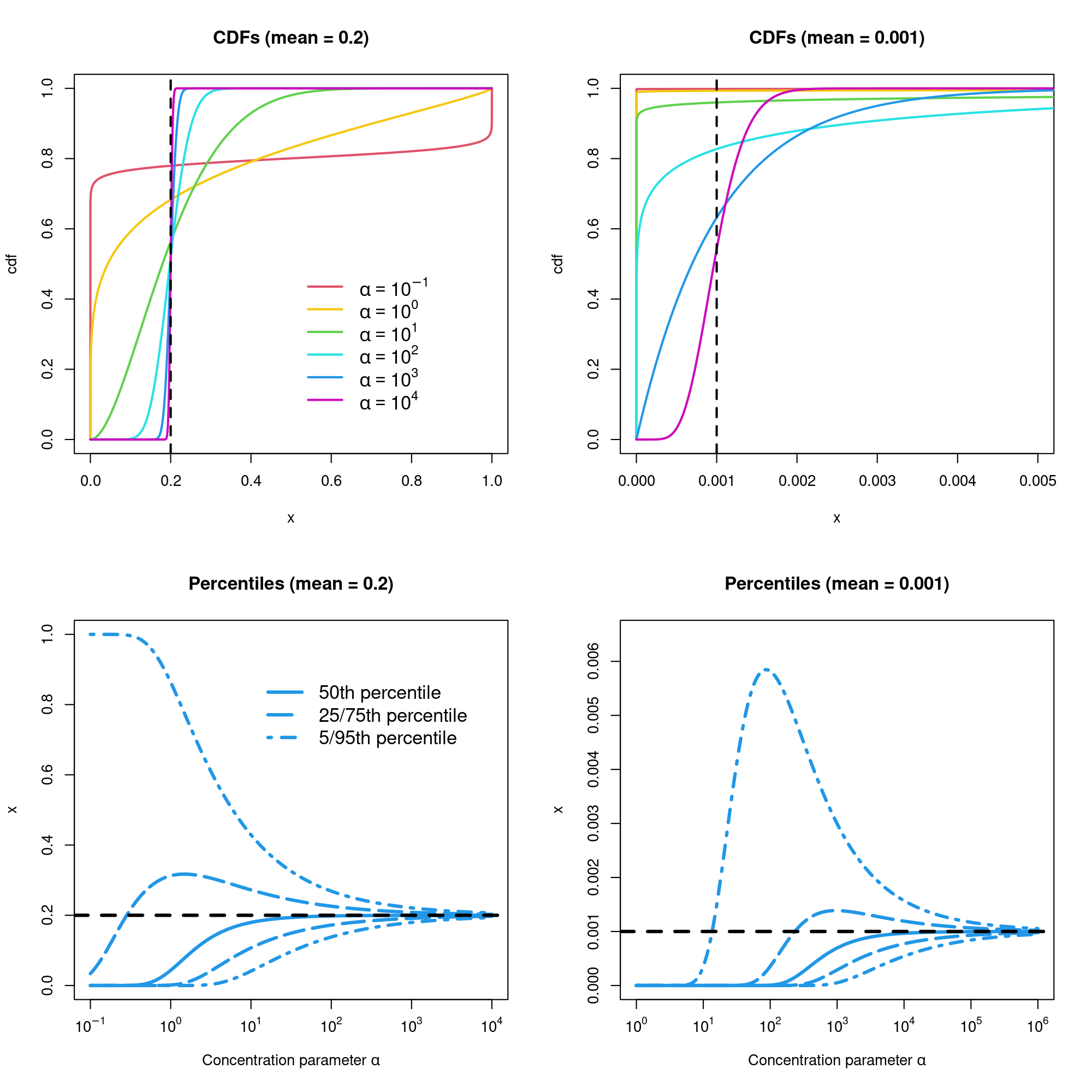}
    \caption{Setting the mean equal to the desired location is problematic near the boundary. (\textit{Top}) Cumulative distribution functions (CDFs) of Beta distributions with means $0.2$ and $0.001$, for a range of concentration parameters $\alpha$. The mean of $\mathrm{Beta}(a,b)$ is $a/(a+b)$ and the concentration parameter is $\alpha = a + b$. (\textit{Bottom}) Percentiles of the same distributions, as a function of $\alpha$.
   Even when the mean is not close to the boundary (such as $0.2$), most of the mass just shifts from the mean to the boundary as $\alpha$ goes from high to low.
    When the mean is close to the boundary (such as $0.001$), the distribution is quite concentrated for all $\alpha$, and $\alpha$ provides very little control over the scale.}
    \label{fig:problem}
\end{figure}
% \cathytodo{After arXiv submission: lengthen the lines in the legend so you can more easily see that 25/75 is dashed and 5/95 is dash-dot.  Also, let's write ``25th & 75th percentiles'' and ``5th & 95th percentiles.''}

This pathological behavior may have severe consequences.  Some Bayesian models using Beta or Dirichlet priors may unintentionally be forcing probabilities to be essentially zero, even when the prior mean is not that close to zero.
Metropolis--Hastings proposals using Beta or Dirichlet distributions may lead to very poor mixing because the proposals are extremely close to the boundary with high probability, rather than being near the current state.
Furthermore, the problem is exacerbated in high dimensions: In a high-dimensional probability simplex, every point is near the boundary because the sum-to-one constraint forces many entries to be close to zero.

In this paper, we propose a novel method for specifying the parameters of a Dirichlet distribution in a way that provides better control over the location and scale.
Specifically, given a target location $c = (c_1,\ldots,c_K)$ and a scale parameter $s$, we maximize the density at $c$ subject to the constraint that a specified measure of scale is equal to $s$.
For instance, the choice of scale may be the concentration parameter, the sum of the variances, the mean cosine error, or some other value quantifying the spread of the distribution.
This \textit{maximum density} approach has several attractive features.
First, it provides greater control over the scale of the distribution.
% -- for instance, one can set $v$ to any value between $0$ and $1/4$, which is the maximum possible variance of a Beta distribution.
Additionally, it tends to put more probability mass near the target location $c$, because it maximizes the density at $c$ by construction.
Furthermore, it can be computed using a fast and simple algorithm for performing the constrained optimization. 

The rest of the article is organized as follows.
In \cref{sec:method}, we describe our proposed methodology.
In \cref{sec:examples}, we illustrate the method's utility for (i) Metropolis--Hastings proposals with improved mixing near the boundary, (ii) Bayesian inference for the probability of rare events,
and (iii) generating random probability vectors for mutational signatures analysis in cancer genomics.
We conclude with a brief discussion in \cref{sec:discussion}.

% \jefftodo{(We need to make sure to describe the algorithm somewhere in the paper, probably in the methodology section.)}

% algorithm to effectively tune the concentration parameter of the Dirichlet distribution so that the same desired concentration around a given mean is maintained regardless of how extreme it is. We are specifically motivated by a cancer genomics application on mutational signature analysis \citep{Alexandrov_2013, Alexandrov_2020}, where an abundance of cancer-specific informative priors is publicly available \aletodo{[Might expand here a little more]}. \cathytodo{In particular, [describe here the method in simple terms]}. We illustrate the advantage of our method in a simple MCMC setting, where \cathytodo{[...]}, and on real data on mutational signatures analysis. 

%%%%%%%%%%%%%%%%%%%%%%%%%%%%%%%%%%%%%%%%%%%%%
% Methods
%%%%%%%%%%%%%%%%%%%%%%%%%%%%%%%%%%%%%%%%%%%%%
\section{Methodology}
\label{sec:method}

In this section, we introduce the maximum density method.
We first consider the special case of Beta distributions (\cref{sec:method-beta}), then generalize to Dirichlet distributions (\cref{sec:method-dirichlet}), and provide a step-by-step algorithm (\cref{sec:method-algorithm}).

\subsection{Maximum density method for specifying Beta distributions}\label{sec:method-beta}

The Beta distribution with parameters $a>0$ and $b>0$, denoted $\mathrm{Beta}(a,b)$, has density $\mathrm{Beta}(x\mid a,b) = x^{a-1} (1-x)^{b-1} / \mathrm{B}(a,b)$ for $x\in(0,1)$, where $\mathrm{B}(a,b)$ is the beta function.
Given a target location $c\in(0,1)$, we propose to set $a$ and $b$ via the optimization:
\begin{align}\label{eq:max_density}
    \argmax_{a,b>0} \mathrm{Beta}(c\mid a,b) \text{ subject to } h(a,b) = 0,
\end{align}
where $h(a,b) = 0$ represents a desired scale constraint.
In the Beta case, we focus particularly on constraining the variance to a given value $v$ by defining $h(a,b) = V(a,b)/v - 1$, where $V(a,b) = a b / \big((a+b)^2 (a+b+1)\big)$ is the variance of $\mathrm{Beta}(a,b)$.
In other words, out of all the Beta distributions with variance $v$, we choose the one with highest density at the target location $c$.
We refer to this as the \textit{maximum density} method.

For $v\in(0,1/4)$, the solution to this optimization always exists and is finite, as we show in \cref{thm:optimization}.
To compute the solution, we develop an algorithm based on Newton's method with equality constraints; see \cref{sec:method-algorithm}.
This algorithm reliably converges to the constrained maximum for any $c\in(0,1)$, $v\in(0,1/4)$.
The convergence is rapid for $v<0.2$, and becomes slower as $v$ approaches $1/4$.

\begin{theorem}
    \label{thm:optimization}
    For all $c\in(0,1)$ and $v\in(0,1/4)$, there exists a finite solution to:
    $$ \argmax_{a,b>0} \mathrm{Beta}(c\mid a,b) \text{ subject to } V(a,b) = v, $$
    where $V(a,b) = a b / \big((a+b)^2 (a+b+1)\big)$.
\end{theorem}

All proofs are provided in \cref{sec:proofs} of the Supplementary Material.  In \cref{thm:optimization}, the restriction to $v\in(0,1/4)$ is necessary since there do not exist Beta distributions with variance greater or equal to $1/4$, as we show in \cref{thm:beta-existence}.

\begin{theorem}
\label{thm:beta-existence}
There is a Beta distribution with mean $u$ and variance $v$ if and only if $0 < v < 1/4$ and $|u - 1/2| < (1/2)\sqrt{1 - 4 v}$.
\end{theorem}

% \jefftodo{(TODO(Jeff): Possible additional theorem: This might be beyond the scope of this version of the paper, but it would also be nice to be able to say that any Beta distribution can be obtained in this way.)}

\subsubsection{Fixing the concentration parameter instead of the variance}

In some situations, it may be preferable to constrain the concentration $\alpha = a+b$ instead of the variance.  In this case, the maximum density method is to set $a$ and $b$ via:
\begin{align}\label{eq:max_density_alpha}
    \argmax_{a,b>0} \mathrm{Beta}(c\mid a,b) \text{ subject to } a+b = \alpha,
\end{align}
given $c\in(0,1)$ and $\alpha > 0$. 
\cref{thm:optimization} extends to this case as well, that is, for all $c\in(0,1)$ and $\alpha\in(0,\infty)$ there is a finite solution to \cref{eq:max_density_alpha}; see \cref{thm:optimization-alpha}.
A slightly different version of the Newton's method algorithm can be used to solve this; see \cref{sec:method-algorithm}. In \cref{fig:method,fig:logit}, we use this fixed $\alpha$ version (rather than fixed variance) to facilitate comparisons with the mean method with fixed $\alpha$.

% \aletodo{This is for a later version of the manuscript: can we find general conditions for the constraint $h$ so that the problem always have finite solutions? This way we write a single general theorem, and then the variance case and the constraint concentration can be put in a corollary. I find it odd to have a theorem for the variance, but a picture for the concentration.} \jefftodo{Yes, this would be nice for a possible future revision.  The way it is now is sort of an artifact of the way the paper evolved.}

\begin{figure}
    \centering
    \includegraphics[trim=0 0 0 0, clip, width=\textwidth]{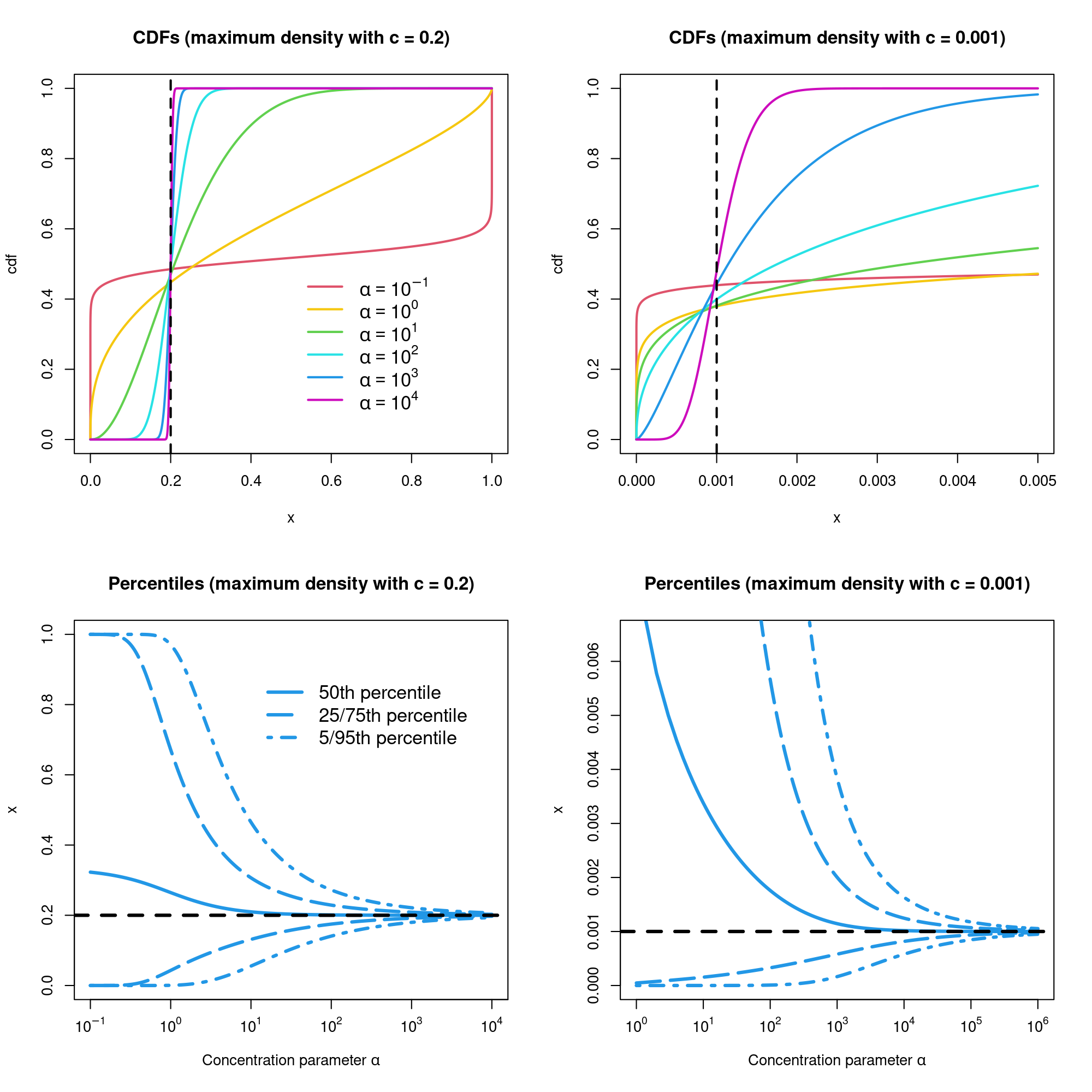}
    \caption{Maximum density provides better control of location and scale near the boundary. (\textit{Top}) CDFs of Beta distributions determined by the maximum density method with target locations $c = 0.2$ and $c = 0.001$, for the same range of concentration parameters $\alpha$ as in \cref{fig:problem}; here, we maximize $\mathrm{Beta}(c \mid a,b)$ subject to $a + b = \alpha$ to enable comparison with \cref{fig:problem}, see \cref{sec:method-algorithm} for details.
    (\textit{Bottom}) Percentiles of the same distributions, as a function of $\alpha$.
    A wide range of scales can be attained by varying $\alpha$ (or correspondingly, varying $v$) when using the maximum density method, regardless of whether the target location $c$ is near the boundary.
    Furthermore, $c$ remains between the first quartile and the median of the distribution as $\alpha$ decreases, rather than moving into the upper tail as in \cref{fig:problem}.}
    \label{fig:method}
\end{figure}

\subsubsection{Comparison of the distributions attained by each method}

The maximum density method resolves the issues with the mean method seen in \cref{fig:problem}.
As shown in \cref{fig:method}, the maximum density method exhibits a greater range of control over the scale of the distribution, while keeping the distribution more centered at the target location in terms of percentiles.
% \cref{fig:median} shows the corresponding plots for the median method, which performs roughly similarly to maximum density in terms of providing better control over the scale.

\begin{figure}
    \centering
    \includegraphics[trim=0.5cm 0 1.25cm 0, clip, width=0.49\textwidth]{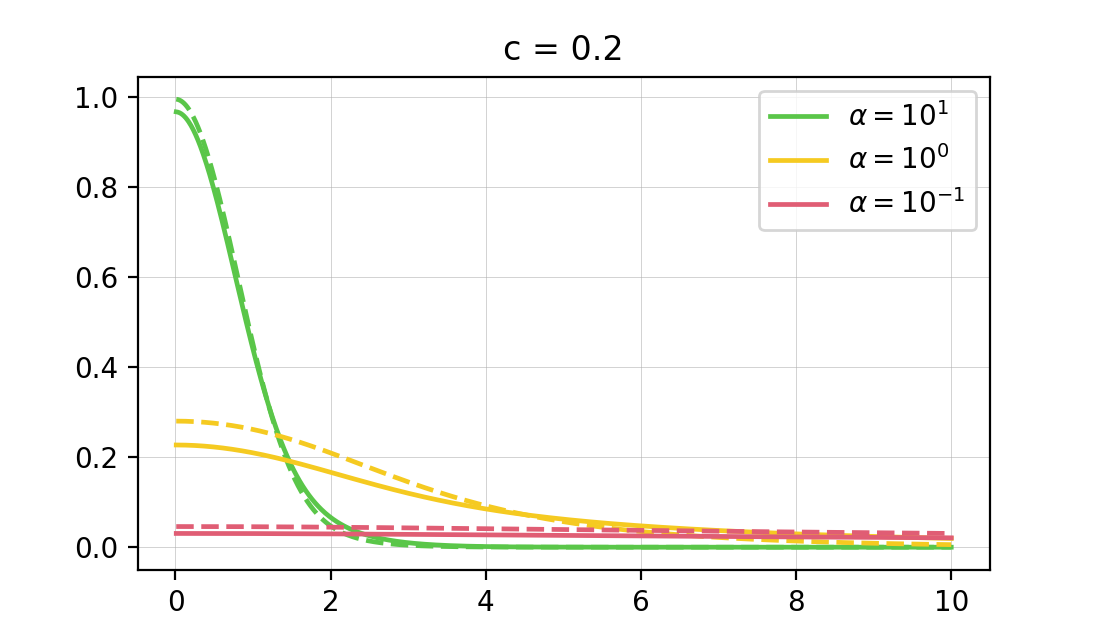}
    \includegraphics[trim=0.5cm 0 1.25cm 0, clip, width=0.49\textwidth]{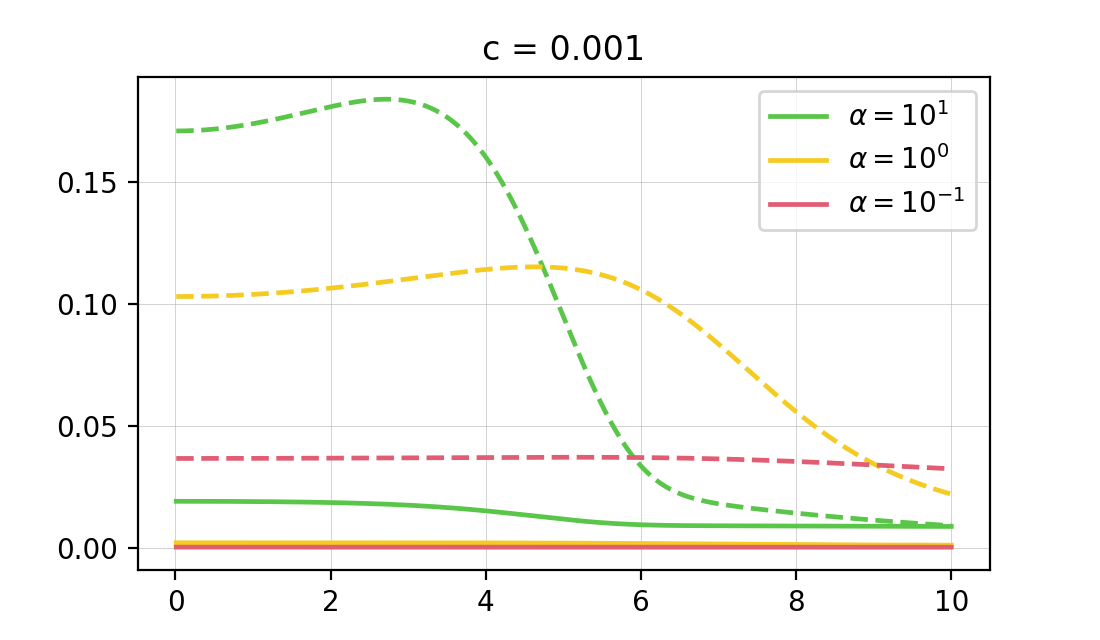}
    \caption{For Beta distributions, the maximum density method puts more probability mass near the target location, compared to the mean method. The plots show the density of $Y = |\mathrm{logit}(X) - \mathrm{logit}(c)|$ when $X$ is Beta distributed with parameters chosen by either (i) setting the mean equal to $c$ (solid line), or (ii) using the maximum density method with target location $c$ (dashed line). For both methods, we use concentration parameters of $\alpha = 10$, $\alpha = 1$, and $\alpha = 0.1$.}
    \label{fig:logit}
\end{figure}

To quantify how far the mass of the distribution is from the target location $c$, we consider 
$Y = |\mathrm{logit}(X) - \mathrm{logit}(c)|$, where
$\mathrm{logit}(x) = \log(x/(1-x))$ and $X\sim\mathrm{Beta}(a,b)$.
Small values of $Y$ mean that $X$ is close to $c$ on the logit scale.
Transforming to the logit scale makes it possible to evaluate differences in magnitude close to the boundary.
\cref{fig:logit} shows the density of $Y$ when using (i) the mean method and (ii) our maximum density method; for both methods we compare results when using $c\in\{0.001,0.2\}$ and $\alpha\in\{0.1,1,10\}$.
More precisely, the mean method chooses $a = \alpha c$ and $b = \alpha (1-c)$.  For the maximum density method, we use target location $c$ and concentration $\alpha$.  We derive a closed-form expression for the density of $Y$ in \cref{sec:logit-derivation}.
% See \cref{fig:logit-median} for a version of this plot that includes the median method as well.

\cref{fig:logit} demonstrates that the maximum density method puts more probability mass near the target location $c$, compared to the mean method.
Specifically, we see that the distribution of $Y$ has more probability mass near $0$ under the maximum density method, meaning that $X$ is closer to $c$ with high probability. The difference is especially stark near the boundary, for instance when $c = 0.001$.
% However, even away from the boundary, such as $c = 0.2$, 

\subsubsection{Median method: An alternative approach}

An alternative to our maximum density approach would be to choose a Beta distribution with median equal to the target location $c$ and variance equal to $v$ (or concentration parameter equal to $\alpha$).
% \jefftodo{(TODO(Jeff): Can we guarantee that such a Beta distribution exists for any $c\in(0,1)$ and any $v\in(0,1/4)$?  We have a partial proof of this.)}
In the case of Beta distributions, this works reasonably well.
However, the maximum density approach has three advantages: 
(1) it tends to put more mass near the target location $c$,
(2) optimization is more tractable and stable,
and (3) it extends more naturally to the general Dirichlet case.
See \cref{sec:median} for more details on this median-based approach.

\subsection{Maximum density method for specifying Dirichlet distributions}
\label{sec:method-dirichlet}

The Dirichlet distribution with parameters $a_1,\ldots,a_K > 0$ has density 
\begin{equation}\label{eq:dir_density}
    \mathrm{Dirichlet}(x \mid a_1,\ldots,a_K) = \frac{\Gamma\big(\sum_{i=1}^K a_i\big)}{\prod_{i=1}^K \Gamma(a_i)}  \, x_1^{a_1-1}\cdots x_K^{a_K-1}
\end{equation}
for $x = (x_1,\ldots,x_K) \in \Delta_K$, 
where $\Delta_K := \big\{x\in\mathbb{R}^K : x_1,\ldots,x_K>0, \, \sum_{i=1}^K x_i=1\big\}$ is the probability simplex.
Here, $\Gamma(x)$ denotes the gamma function.

Letting $X\sim\mathrm{Dirichlet}(a_1,\ldots,a_K)$, it can be shown that the $i$th coordinate $X_i$ follows a Beta distribution, specifically, $X_i \sim \mathrm{Beta}(a_i, \, \sum_{j\neq i} a_j)$.
Thus, the Beta distribution can essentially be thought of as the special case of a Dirichlet with $K = 2$, although technically the Beta corresponds to one coordinate of a Dirichlet.
The mean of the $i$th coordinate is $\E(X_i) = a_i/\sum_{j=1}^K a_j$ and the scale of the Dirichlet distribution around its mean is traditionally thought to be controlled by the \textit{concentration parameter}, $\alpha = \sum_{i=1}^K a_i$. 
% \aletodo{This part should be made a little more explicit in the introduction} \jefftodo{Done.}
However, when the mean of $X_i$ is near zero or one, the distribution of $X_i$ exhibits the same pathologies as in the case of a Beta distribution.
This follows simply because $X_i$ is, in fact, Beta distributed.
In particular, the concentration parameter $\alpha$ exhibits little control over the scale, and the distribution of $X_i$ concentrates near the boundary as $\alpha$ decreases.

We extend our maximum density method to the general case of a $K$-dimensional Dirichlet distribution as follows.
Suppose $c = (c_1,\ldots,c_K)\in\Delta_K$ is a target location in the probability simplex.
We propose to choose the Dirichlet parameters $a_1,\ldots,a_K$ via:
\begin{align}\label{eq:max_density_dirichlet}
    \argmax_{a_1,\ldots,a_K>0} \mathrm{Dirichlet}(c\mid a_1,\ldots,a_K) \text{ subject to } h(a_1,\ldots,a_K) = 0,
\end{align}
where $h$ represents a desired constraint.  This is a direct generalization of \cref{eq:max_density}.
We provide an algorithm for solving this optimization problem in \cref{alg:dirichlet}.
One simple choice of constraint is to fix the concentration parameter to a given value $\alpha$; we do this by defining $h(a_1,\ldots,a_K) = (\sum_{i=1}^K a_i)/\alpha - 1$.
Another option would be to control the variances, however, in the multivariate case we have to summarize the scale of the distribution with a single number.
Motivated by an application to mutational signatures, we consider controlling the mean cosine error between $X$ and $\mathbb{E}(X)$; see Sections~\ref{sec:mutational-signatures} and \ref{sec:derivation-of-algorithm} for details.

\cref{fig:logit-dirichlet} shows that for Dirichlet distributions, the maximum density method puts more mass near $c$ compared to setting the mean equal to $c$; compare with \cref{fig:logit} in the Beta case.  Here, we constrain $\alpha$ to enable direct comparison with the mean method.

% Since we are primarily concerned with issues near the boundary, it is natural to focus on the coordinates with the most extreme values of $c_i$.
% Since $\sum_{i=1}^K c_i = 1$, at least one of the values of $c_i$ nearest to $0$ or $1$ equals $\min_i c_i$. 
% Thus, to control the variance, we propose to define
% $h(a_1,\ldots,a_K) = V(a_m, \, \sum_{i\neq m} a_i) / v - 1$
% where $m \in \arg \min_i c_i$ and $v\in(0,1/4)$ is a given variance to be matched.
% Here, $V(a,b) = a b / \big((a+b)^2 (a+b+1)\big)$ is the variance of $\mathrm{Beta}(a,b)$,
% and thus, $\mathrm{Var}(X_m) = V(a_m, \, \sum_{i\neq m} a_i)$ when $X\sim\mathrm{Dirichlet}(a_1,\ldots,a_K)$.
% \jefftodo{(I don't like this way of doing the variance version with the Newton algorithm, because it is not invariant to the choice of argmin.
% The only real virtue is that it is a generalization of the Beta version.  Another generalization would be to constrain $(a_1 \cdots a_K) / ((\sum_i a_i)^K (\sum_i a_i + 1))$, but I don't know if this is a good idea or if it has a clear interpretation.)}

\begin{figure}
    \centering
    \includegraphics[trim=0 0 0.5cm 0, clip, width=0.49\textwidth]{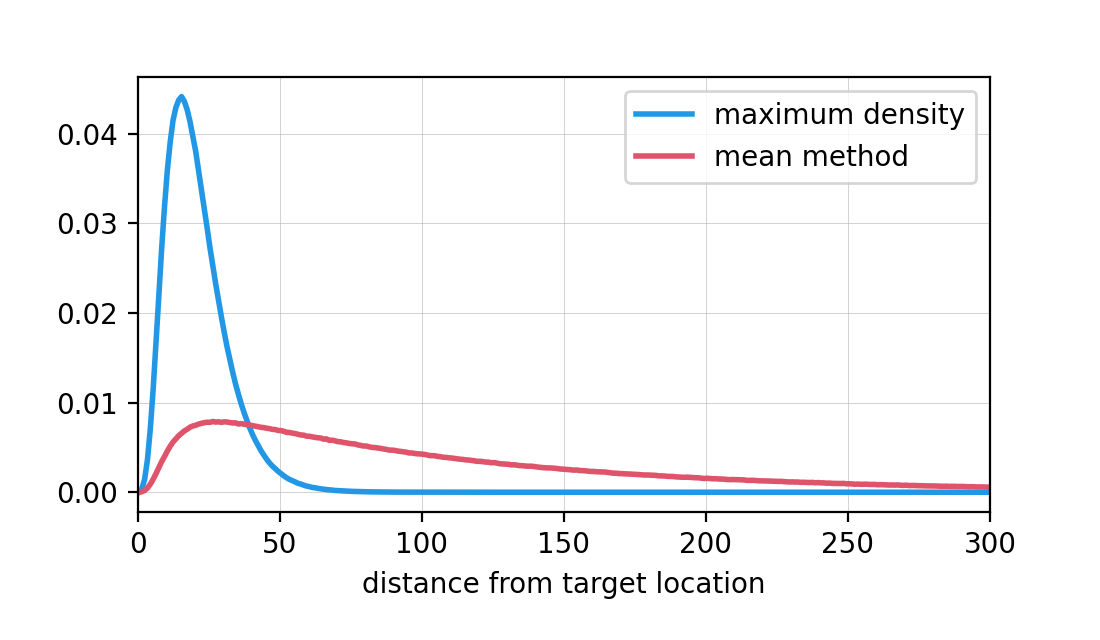}
    \includegraphics[trim=0 0 0.5cm 0, clip, width=0.49\textwidth]{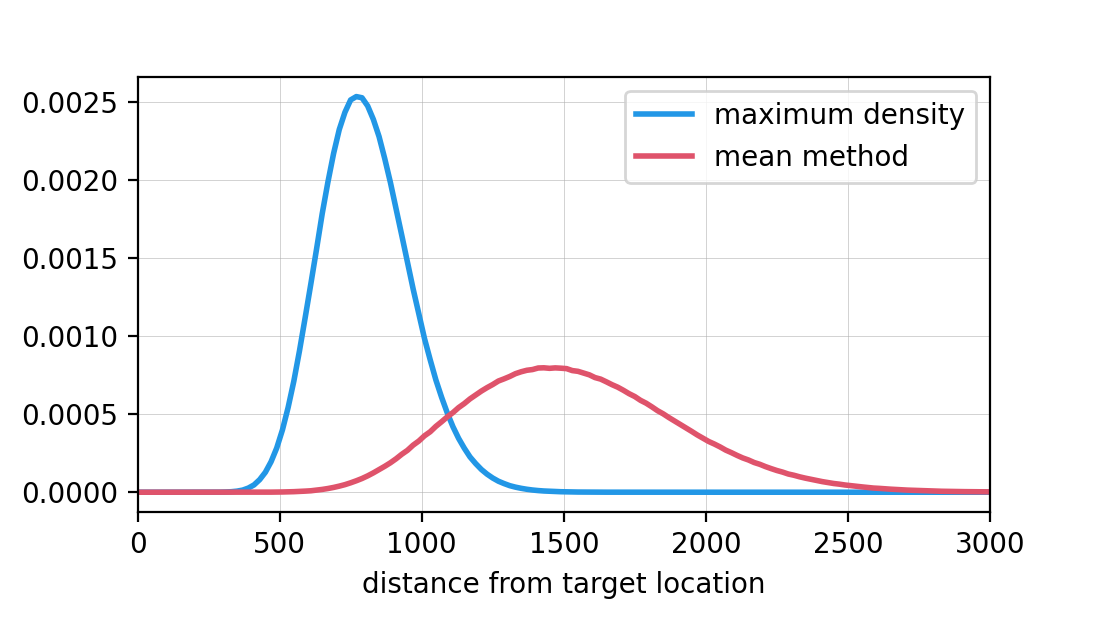}
    \caption{For Dirichlet distributions, the maximum density method puts more probability mass closer the target location. The plots show the density of $Y = \sum_{i=1}^K |\mathrm{logit}(X_i) - \mathrm{logit}(c_i)|$ when $X\sim\mathrm{Dirichlet}(a_1,\ldots,a_K)$ with $a_1,\ldots,a_K$ chosen by either (i) setting the mean equal to $c$ (red line), or (ii) using the maximum density method with target location $c$ (blue line). The concentration parameter is $\alpha = 1$ for both methods. (\textit{Left}) $c = (0.01, 0.1, 0.2, 0.3, 0.39)$; (\textit{Right}) $c = (1, 2, 3, \ldots, 30) / \sum_{i=1}^{30} i$.}
    \label{fig:logit-dirichlet}
\end{figure}

\subsection{Optimization algorithm for the maximum density method}
\label{sec:method-algorithm}

To solve the constrained optimization problem for the maximum density method (\cref{eq:max_density,eq:max_density_dirichlet}), we use an algorithm based on Newton's method with equality constraints \citep{boyd2004convex}.
This generalization of Newton's method is a second-order optimization technique for problems of the form 
\begin{align*}
    \argmin_{x} f(x) \text{ subject to } h(x) = 0,
\end{align*}
where $f$ is a strictly convex, twice-differentiable function and $h$ is differentiable. 
We provide versions of this algorithm for implementing the maximum density method with various choices of constraint function; see \cref{sec:derivation-of-algorithm} for the derivation.

\begin{algorithm}[t]
\caption{~~Maximum density method for Dirichlet distributions}\label{alg:dirichlet}
\begin{algorithmic}[1]
\Require{target location $c\in\Delta_K$, step size $\rho\in(0,1]$, maximum number of iterations $\texttt{maxiter}>0$, and convergence tolerance $\texttt{tol}>0$.}
\State $a_i \gets 10(c_i + 1)/2$ for $i\in\{1,\ldots,K\}$ \Comment{Initialization}
\For{$\texttt{iter} = 1,\ldots,\texttt{maxiter}$}
    \State $s \gets a_1 + \cdots + a_K$ 
    \State $g_i \gets \psi(a_i) - \psi(s) - \log(c_i)$ for $i\in\{1,\ldots,K\}$ \Comment{Compute gradient}
    \State $H_{i j} \gets \psi'(a_i)\mathds{1}(i=j) - \psi'(s)$ for $i,j\in\{1,\ldots,K\}$\Comment{Compute Hessian matrix}
    \State $h \gets h(a)$ \Comment{Value of constraint}
    \State $J_i \gets \partial h/\partial a_i$ for $i\in\{1,\ldots,K\}$ \Comment{Jacobian of constraint}
    \State Solve for $\delta\in\mathbb{R}^K$ and $\lambda\in\mathbb{R}$ in the linear system: \Comment{Compute Newton step}
        \begin{align*}
        \begin{bmatrix}
            H & J^\mathtt{T} \\
            J & 0
        \end{bmatrix}
        \begin{bmatrix}
            \delta \\ \lambda
        \end{bmatrix}
        =
        \begin{bmatrix}
            -g \\ -h
        \end{bmatrix}
        \end{align*}
    ~~~~where $H = [H_{i j}]\in\mathbb{R}^{K\times K}$, $J = [J_1,\ldots,J_K] \in\mathbb{R}^{1\times K}$, and $g = [g_i]\in\mathbb{R}^{K\times 1}$.
    \State $a' \gets a$ \Comment{Store current values}
    \State $a \gets a' + \rho\,\delta$ \Comment{Update values}
    \State for $i\in\{1,\ldots,K\}$, if $a_i \leq 0$ then $a_i \gets a_i'/2$  \Comment{Enforce boundary constraints}
    \If{$|h| + \sum_{i=1}^K |a_i/a_i' - 1| < \texttt{tol}$} \Comment{Check for convergence}
    \State \textbf{output} $a$ \Comment{Return output}
    \EndIf
\EndFor
\Ensure{Dirichlet parameters $a = (a_1,\ldots,a_K)$.}
\end{algorithmic}
\end{algorithm}

\cref{alg:dirichlet} provides a step-by-step procedure for the general case of Dirichlet distributions; see \cref{alg:beta} for the special case of Beta distributions.
As default settings, we use $\rho = 1/2$, $\texttt{maxiter} = 100$, and $\texttt{tol} = 10^{-8}$.
If \cref{alg:dirichlet} reaches $\texttt{maxiter}$ iterations without converging, then we set $\rho \gets \rho/5$ and $\texttt{maxiter} \gets 5\,\texttt{maxiter}$, and run the algorithm again; if it still fails to converge after 5 such restarts, we stop.
In \cref{alg:dirichlet}, $\psi(x)$ and $\psi'(x)$ are the digamma and trigamma functions, that is, the first and second derivatives of $\log\Gamma(x)$.
Also, $\mathds{1}(\cdot)$ is the indicator function, that is, $\mathds{1}(E) = 1$ if $E$ is true, and $\mathds{1}(E) = 0$ otherwise. 

In lines 6--7 of \cref{alg:dirichlet}, $h(a) = h(a_1,\ldots,a_K)$ is the chosen constraint function in \cref{eq:max_density_dirichlet} and $\partial h/\partial a_i$ is its partial derivative with respect to $a_i$.  
For example, to constrain the concentration parameter to equal $\alpha$, line 6 becomes 
$h \gets s/\alpha - 1$
and line 7 becomes $J_i \gets 1/\alpha$ for $i\in\{1,\ldots,K\}$.
See \cref{sec:derivation-of-algorithm} for more details and formulas for handling other constraints.

\section{Examples}
\label{sec:examples}

\subsection{Metropolis--Hastings proposals for probabilities}
\label{sec:mh-example}

Markov chain Monte Carlo (MCMC) is commonly used for posterior inference in Bayesian models.
For models containing a latent vector of probabilities, say $x = (x_1,\ldots,x_K)$, it is necessary to construct MCMC moves on the probability simplex. Except in cases where the prior on $x$ is Dirichlet and the corresponding likelihood is multinomial, the full conditional distribution of $x$ will rarely have a closed form that can be sampled from.
In such cases, a typical approach is to use the Metropolis--Hastings (MH) algorithm to perform an MCMC move that preserves the full conditional distribution of $x$.
A natural choice of MH proposal distribution on the probability simplex is a Dirichlet with mean equal to the current state of $x$.
% we refer to this as the mean method.

However, when the target distribution of $x$ (for instance, the full conditional) places non-negligible mass near the boundary of the simplex, this is highly suboptimal due to the issues illustrated in \cref{fig:problem}.
Specifically, when the current state of $x$ has one or more entries $x_i$ that are near $0$, a Dirichlet proposal with mean $x$ will be concentrated either (i) very near $x$ itself or (ii) extremely close to the boundary.  In case (i) any moves will be very small, and in case (ii) there will be a very small probability of moving closer to the center of the simplex relative to $x$.
% \aletodo{too many repeated ''Dirichlet proposal''.} \jefftodo{fixed}
As a result, the MCMC sampler will have difficulty moving around the space efficiently.

The maximum density method can be used to construct MH proposal distributions that yield better MCMC mixing.
Specifically, we propose to set the target location $c$ to be the current state of $x$ in the MCMC sampler, set the scale $s$ to a preselected value (for instance, based on pilot runs of the sampler to tune $s$ for good performance), use \cref{alg:dirichlet} to choose the Dirichlet parameters $a_1,\ldots,a_K$, and then use $\mathrm{Dirichlet}(a_1,\ldots,a_K)$ as the MH proposal distribution.  The numerator of the MH acceptance ratio can be computed by using the same procedure to define the proposal distribution at the proposed value; see \cref{sec:mh-details} for details.  This provides better control over the location and scale of the proposals, improving MCMC performance.

To illustrate, we compare several methods for constructing Beta-distributed MH proposals on the unit interval $(0,1)$. Consider the following four target distributions:
\begin{enumerate}[(A)]
	\item \textit{Uniform:} $\mathrm{Beta}(1,1)$
	\item \textit{Unimodal at zero:} $\mathrm{Beta}(1,1000)$
	\item \textit{Bimodal mixture:} $0.75 \, \mathrm{Beta}(2,5) + 0.25\, \mathrm{Beta}(10,2)$ 
	\item \textit{Bimodal at zero and one:} $\mathrm{Beta}(1/2, 1/2)$.
\end{enumerate}
We compare the performance of the following methods.  Letting $x$ denote the current state, consider using an MH proposal consisting of a Beta distribution with:
\begin{enumerate}[(I)]
	\item maximum density for target location $x$ and fixed variance $v = 0.1$,
	\item mean $x$ and fixed concentration parameter $\alpha = 5$,
	\item mean $x$ and fixed variance $v = 0.1$, or
	\item mean $x$ and standard deviation $\sigma = \min\{x, 1-x, \sqrt{0.1}\}$.
\end{enumerate}

The choices of $v = 0.1$ and $\alpha = 5$ are based on pilot runs with a range of values to determine choices that perform well; see \cref{fig:mcmc-tuning-alpha,fig:mcmc-tuning-max-density}.
Method IV is motivated by the observation that if $x$ is close to $0$, then choosing mean $x$ and $\sigma \approx x$ keeps the mass of the distribution relatively near $x$; likewise for $\sigma \approx 1-x$ when $x$ is close to $1$. Keeping $\sigma \leq \sqrt{0.1}$ avoids issues with non-existence of Beta distributions with large variance, as characterized by \cref{thm:beta-existence}.
We refer to method IV as using ``adaptive variance.''
Method III appears to be reasonable at first, but is fundamentally flawed since by \cref{thm:beta-existence} there does not always exist a Beta distribution with mean $x\in(0,1)$ and variance $v = 0.1$. To make Method III well defined, we reject any proposal to a value of $x$ for which a return is impossible due to non-existence of the proposal distribution.

For each target distribution, we run each MCMC sampler for $10{,}000$ iterations after a burn-in of $100$ iterations; see \cref{sec:mh-details} for a detailed description of the sampler.  This is repeated $100$ times for each combination of target distribution (A-D) and method (I-IV).
To evaluate performance, we consider (i) the autocorrelation function, to quantify mixing performance, and (ii) the Kolmogorov--Smirnov distance between the posterior samples and the target distribution, to verify convergence to the target.

\begin{figure}
    \centering
    \includegraphics[trim=0.5cm 0 0.25cm 0, clip, width=1\textwidth]{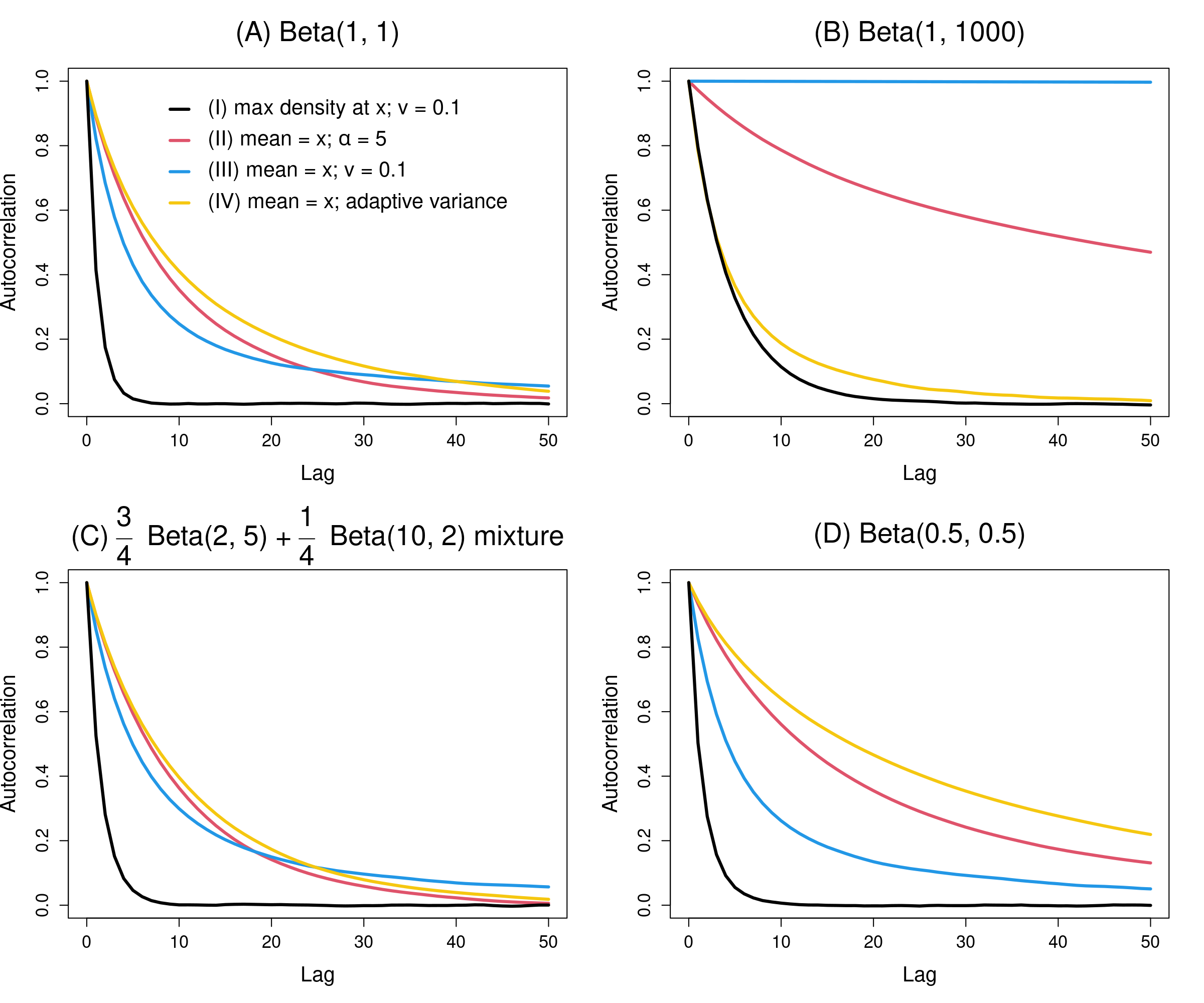}
    \caption{The maximum density method yields MH proposals with improved MCMC mixing compared to the mean method. Each plot shows the autocorrelation functions (ACFs) for one of the four target distributions (A-D). ACFs are shown for each method of proposal distribution specification (I-IV). Maximum density has lower autocorrelation than the mean methods, indicating that the sampler exhibits better performance.}
    \label{fig:autocorrelation}
\end{figure}

\cref{fig:autocorrelation} shows the estimated autocorrelation functions (ACFs) for each combination of target distribution (A-D) and method (I-IV). 
Compared to the ACFs for the mean-based methods, the ACF for our maximum density method decays significantly faster, indicating that the sampler is more efficiently traversing the target distribution.
In \cref{fig:autocorrelation}, method III (mean $x$, variance $0.1$) appears to perform reasonably well on target distributions A, C, and D, but this is misleading.
In fact, method III is not even converging to the target distribution, as we demonstrate next.

\begin{figure}
    \centering
    \includegraphics[trim=0 0 0.25cm 2.5in, clip, width=0.8\textwidth]{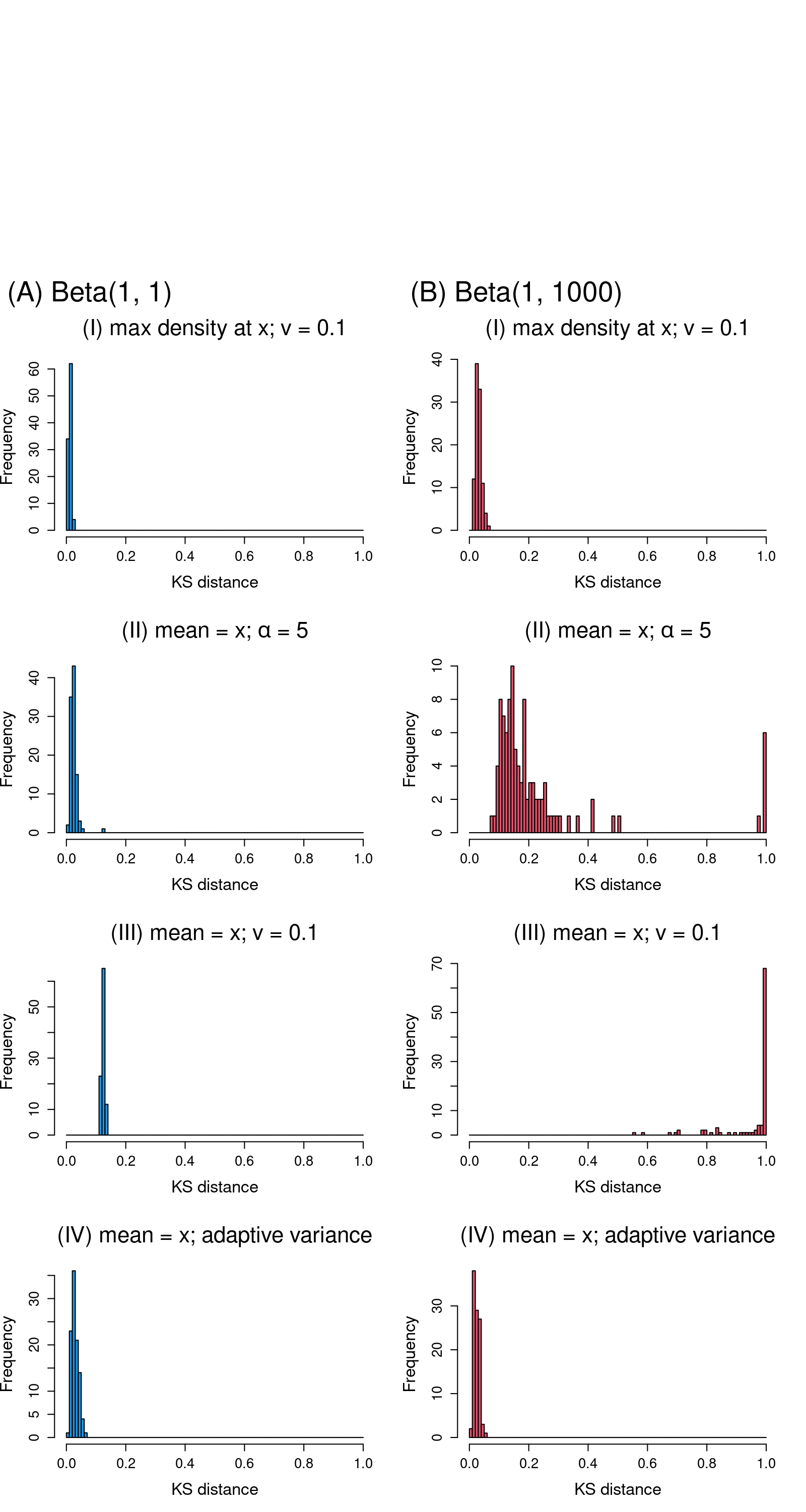}
    \caption{Histograms of the Kolmogorov--Smirnov distance between the target distribution and the MCMC approximation to it. 
    The blue histograms in the left column are for the $\mathrm{Beta}(1,1)$ target distribution (case A); the red histograms in the right column are for the $\mathrm{Beta}(1,1000)$ target distribution (case B).
    Samplers using maximum density proposals (method I) or the mean method with adaptive variance (method IV) converge reliably to the target distribution.  
    Method III fails even for simple target distributions.
    Method II fails when the target distribution puts significant mass near the boundary.
    % \cathytodo{(Cathy: Make plot for the other two target distributions for supplement.)}
    }
    \label{fig:ks-distance}
\end{figure}

\cref{fig:ks-distance} shows the Kolmogorov--Smirnov (KS) distances between the target distribution and the MCMC approximation based on samplers.
For each combination of target distribution and method, the figure shows the distribution of KS distances over $100$ replicate runs of the MCMC sampler.
For the $\mathrm{Beta}(1,1)$ target distribution, all methods except III appear to successfully converge to the target distribution.
% Meanwhile, method III is producing samples for which the KS distance appears to be bounded away from zero.
The reason why method III (mean $x$, variance $0.1$) fails to converge to the target distribution is because
% there are values of $x\in(0,1)$ for which this proposal distribution does not exist. Specifically,
by \cref{thm:beta-existence}, there exists a Beta with mean $x$ and variance $v$ if and only if $|x - 1/2| < (1/2)\sqrt{1-4 v}$. Consequently, the sampler for Method III cannot reach any point outside this interval.

For the $\mathrm{Beta}(1,1000)$ target distribution, which is more challenging since it is concentrated near $0$, \cref{fig:ks-distance} indicates that only methods I (maximum density) and IV (mean $x$ with adaptive variance) successfully converge to the target distribution.
Here, even method II (mean $x$, fixed $\alpha=5$) fails to converge within the allotted number of MCMC iterations.  Since method II is valid, it should eventually converge, but it may take a very large number of iterations.
Method III fails again since it is invalid.

Overall, these results show that our maximum density method yields MH samplers that successfully converge to the target distribution (\cref{fig:ks-distance}) and do so much more rapidly than the mean-based methods (\cref{fig:autocorrelation}).

% We find that $v = 0.1$ is a reasonable default choice of proposal variance, but in general $v$ would need to be chosen based on the scale of the target distribution.
% Smaller proposal variances lead to higher acceptance probabilities, but slower mixing if the variance of the target distribution is actually much larger.

\subsection{Bayesian modeling of rare events} 

Using a good choice of prior is particularly important when performing Bayesian inference for the probability of rare events.  Suppose the true probability of the events is on the order of $1 / n$ or smaller, where $n$ is the number of observations. In such situations, the prior distribution can have a strong influence on posterior inferences. 

Consider a simple Bernoulli model: $Y_1,\ldots,Y_n$ i.i.d.\ $\sim\mathrm{Bernoulli}(\theta)$, where $Y_i$ is the observed binary outcome of event $i$ and $\theta$ is unknown.
A typical choice of prior on $\theta$ would be a Beta distribution with mean equal to a location $c\in(0,1)$ that one expects $\theta$ to be near, \emph{a priori}.
To illustrate, suppose we expect $\theta$ to be around $c = 10^{-3}$, and we use a prior of $\theta \sim \mathrm{Beta}(\alpha c, \alpha(1-c))$ with concentration parameter $\alpha = 10$, so that the prior mean is $c$.
Now, suppose the true parameter value is $\theta_0$ and the true data generating process is  $Y_1,\ldots,Y_n$ i.i.d.\ $\sim\mathrm{Bernoulli}(\theta_0)$, where $n = 100$.

One might hope that $\theta_0$ would be a typical value under the posterior distribution of $\theta$.
For instance, if the posterior is appropriately quantifying uncertainty about the true value, then 95\% credible intervals should contain $\theta_0$ most of the time -- ideally, around 95\% of the time for correct frequentist calibration.
% Equivalently, the value of the posterior CDF at $\theta_0$, namely,  $F_n(\theta_0) := \mathbb{P}(\theta < \theta_0 \mid Y_1,\ldots,Y_n)$, should be contained in the interval $(0.025,0.975)$ most of the time -- ideally, around 95\% of the time for correct calibration.
However, \cref{fig:binomial-ex-fixed-target-location} (left) shows that this is not the case for the mean method. 
The coverage of 95\% highest posterior density (HPD) intervals is very low for a substantial range of $\theta_0$ values. Even at $\theta_0 = c = 10^{-3}$ (dotted line), where the true parameter equals the mean of the prior distribution, the coverage is only around 10\%. 
The reason is that the prior is concentrated near $0$ rather than around $c$, as illustrated in \cref{fig:problem}.
This also explains why the coverage jumps up to 100\% for very small values of $\theta_0$ (less than $\approx 10^{-4.5}$).
Thus, the mean method performs worst in the range of $\theta_0$ values where we want it to perform the best (near $c$).

% For smaller $\theta_0 < 1 / n$, we do not expect to observe any positive Bernoulli trials and the posterior distribution of $\theta$ suffers from the skewness of the prior (\cref{fig:problem}). \jefftodo{(Does this explanation make sense? One possible thing to add is that the dropoff happens right where the quantile of $\theta_0$ in $\mathrm{Beta} \left( \alpha c, \alpha c + n \right)$ (the posterior when $\sum Y_i = 0$) falls below 0.95.)}

We propose to instead use our maximum density method to choose the prior on $\theta$.  Specifically, we consider using a $\mathrm{Beta}(a,b)$ prior with $a$ and $b$ obtained via \cref{eq:max_density_alpha} with target location $c = 10^{-3}$ and concentration parameter $\alpha = 10$.
\cref{fig:binomial-ex-fixed-target-location} (left) shows that the resulting posteriors are better calibrated, in the sense that $\theta_0$ tends to fall within the 95\% HPD interval over a very wide range of $\theta_0$ values, even when the true parameter $\theta_0$ is not particularly close to the prior target location $c = 10^{-3}$. The coverage for both methods drops as $\theta_0$ approaches $1$, which makes sense since the prior location of $c = 10^{-3}$ is badly misspecified when $\theta_0$ is close to $1$. 

Even when the prior is centered at the true parameter, the mean method fails.
Suppose the prior is $\theta\sim\mathrm{Beta}(\alpha\theta_0,\alpha(1-\theta_0))$, so that the prior mean equals the true parameter $\theta_0$.
\cref{fig:binomial-ex-fixed-target-location} (right) shows that the resulting coverage is low for all values of $\theta_0$ less than around $10^{-2}$.
Meanwhile, specifying the prior using our maximum density approach with target location $c = \theta_0$ and concentration parameter $\alpha$ yields high coverage for all values of $\theta_0$; see \cref{fig:binomial-ex-fixed-target-location} (right). We use $\alpha = 10$ for both methods.
Of course, it is unrealistic to make the prior centered at the true value; the point is that even in this ideal situation, the mean method still fails.
In contrast, the maximum density method works well---not only in this ideal situation---but also in the realistic situation with a fixed prior that does not depend on $\theta_0$, as shown in \cref{fig:binomial-ex-fixed-target-location} (left).

% We simulate $n = 20$ iid Bernoulli trials for a range of probabilities $p$ from $10^{-10}$ to $10^{-1}$, with 1000 replicates for each value of $p$. We consider three different Beta priors ``centered'' at a prior point estimate of $c = 10^{-5}$: \textsf{na\"ive}, the Beta distribution with mean $c$ and concentration parameter $\alpha = 10$; \textsf{median}, the Beta distribution with median $c$ and variance $V = 0.01$; and \textsf{density maximization}, the Beta distribution with variance $V = 0.01$ that maximizes density at $c = 10^{-5}$.

\begin{figure}
    \centering
    \includegraphics[width=0.45\textwidth]{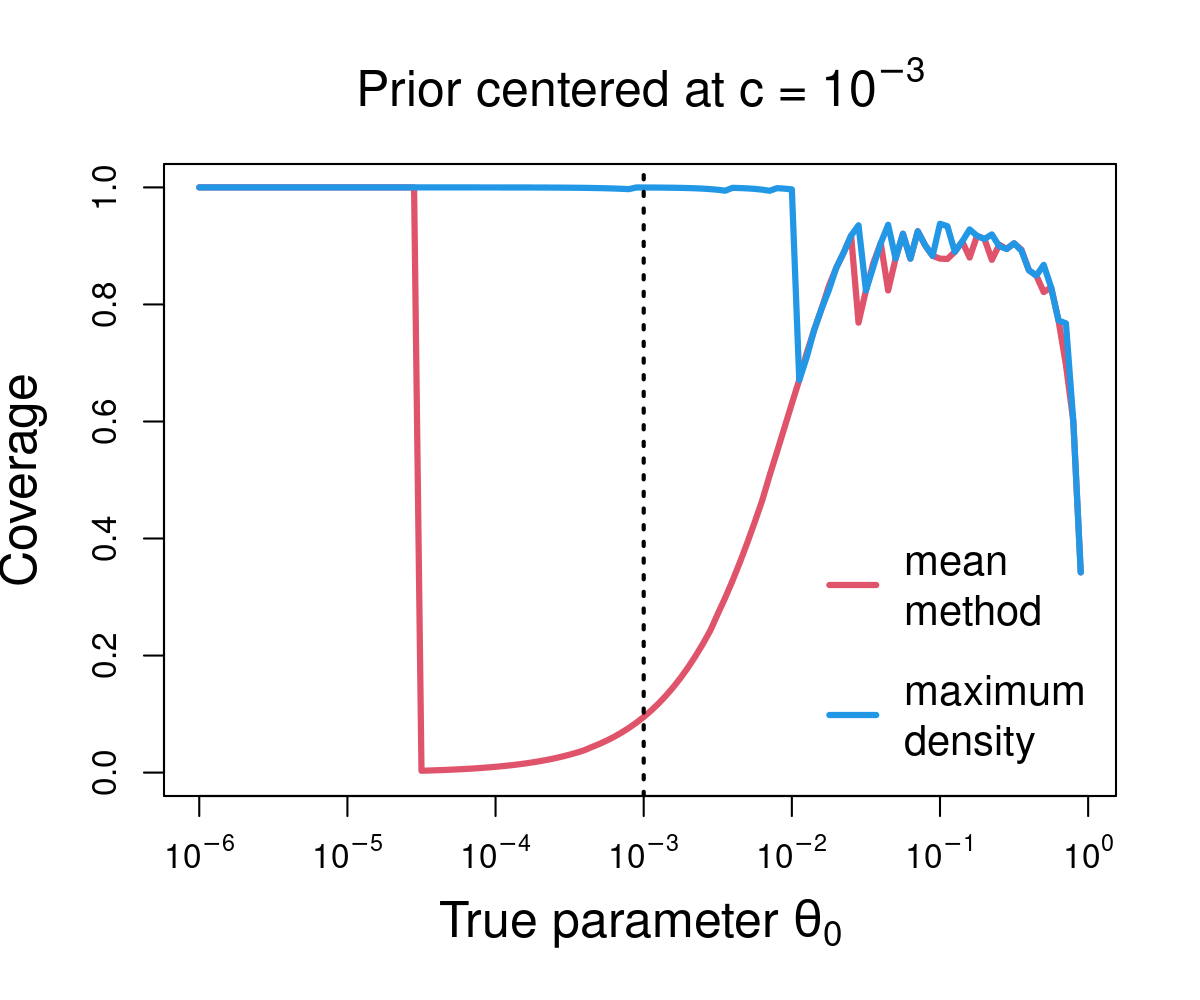}
    \includegraphics[width=0.45\textwidth]{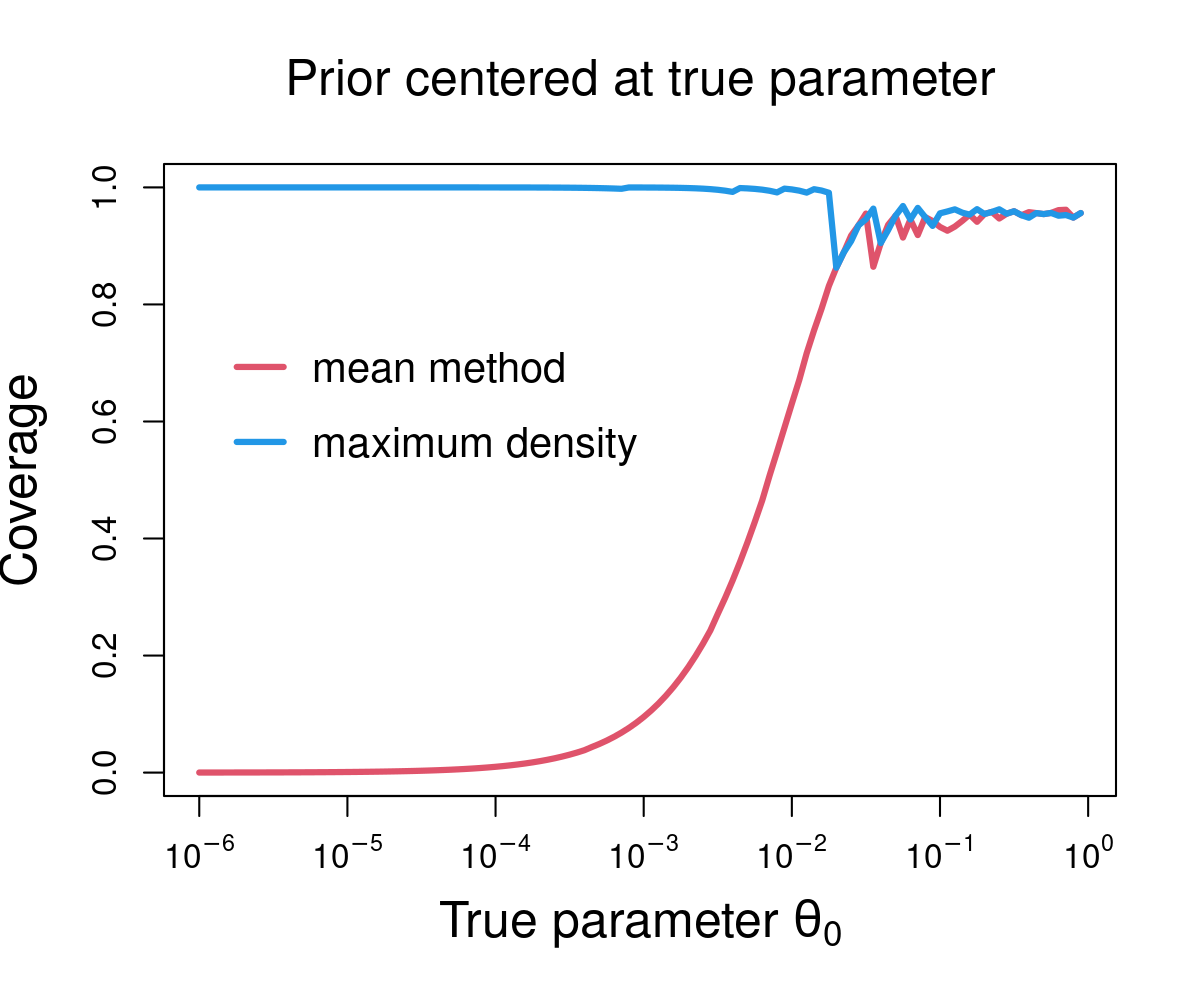}
    \caption{(\textit{Left}) Coverage of 95\% highest posterior density intervals as a function of the true parameter $\theta_0$, using Beta priors with concentration $\alpha=10$ and target location $c = 10^{-3}$ (vertical dotted line). 
    Note: The jagged nature of these lines is due to the discreteness of the distribution, not due to approximation error; the plots show the exact coverage.
    (\textit{Right}) Same as left, but using Beta priors with target location equal to the true parameter ($c = \theta_0$). The mean method performs very poorly for a wide range of small $\theta_0$ values, whereas our maximum density method performs well across the board.
    }
    \label{fig:binomial-ex-fixed-target-location}
\end{figure}

% In contrast to the median and maximum density priors, using the na\"ive approach of setting the target location as the mean of the Beta prior only reasonably captures the true parameter $p$ in the posterior distribution on $p$ for $p = 0.1$, and not for any smaller values for the true parameter (\cref{fig:binomial-ex-fixed-target-location}). This is true even when the prior mean matches the true parameter at $p = 10^{-5}$ (dotted line). In an idealized scenario where the priors are correctly centered at the true parameter for all values of $p$, the posterior distribution on $p$ still fails to adequately capture the true parameter when using the na\"ive prior while the  (\cref{fig:binomial-ex-target-true}). 

\subsection{Simulating random mutational signatures in cancer genomics}
\label{sec:mutational-signatures}

The set of mutations in a cancer genome represents the cumulative effect of numerous mutational processes such as environmental exposures and dysregulated cellular mechanisms.
It turns out that each mutational process tends to produce each type of mutation at a relatively constant rate, and these rates can be represented as a probability vector $c = (c_1,\ldots,c_K)$ referred to the corresponding ``mutational signature'' 
\citep{Nik_zainal_2012, Nik-Zainal_2016,Alexandrov_2013, Alexandrov_2020}.
Here, $c_i$ represents the rate at which mutation type $i$ occurs for the mutational process under consideration.
Usually, one considers the $K=96$ types of single-base substitution (SBS) mutations; see \cref{fig:mutational-signatures} (left) for examples.
The study of mutational signatures has been instrumental in advancing cancer research \citep{Koh_2021,Aguirre_2018, Rubanova_2020}.  % Gulhan_2019

The Catalogue of Somatic Mutations in Cancer \citep[COSMIC;][]{Alexandrov_2020} publish a curated collection of signatures based on thousands of cancer genomes from a wide range of cancer types.
COSMIC signatures are widely used in cancer genome analysis, but it can be important to allow for departures from the COSMIC signatures due to cancer-specific or subject-specific variation \citep{Degasperi_2020, Zou_2021}.
This variation can be represented by a Dirichlet distribution centered at the signature of interest $c$ \citep{zito2024compressivebayesiannonnegativematrix}. % todo: Also cite BayesPowerNMF once posted to arXiv!
However, if one uses $\mathrm{Dirichlet}(\alpha c_1, \ldots, \alpha c_K)$ where $\alpha > 0$ is the concentration, then the variability around $c$ depends strongly on the sparsity of $c$.  The standard measure of the discrepancy between two mutational signatures, say $x$ and $c$, is the cosine error, 
\begin{align}
    \mathrm{CosErr}(x,c) = 1 - \frac{x^\mathtt{T} c}{\|x\| \|c\|}
\end{align}
for $x,c\in\Delta_K$, where $\|x\| = \sqrt{\sum_i x_i^2}$ is the Euclidean norm.
\cref{fig:mutational-signatures} (right) shows that $\alpha$ provides poor control over the mean cosine error between a random signature $X\sim \mathrm{Dirichlet}(\alpha c_1, \ldots, \alpha c_K)$ and its mean $c$.  For any given value of $\alpha$, the mean cosine error can take a very wide range of values depending on the signature $c$.
Consequently, $\alpha$ does not represent the scale of variability around COSMIC signatures in a consistent way across signatures, when using the mean method.

\begin{figure}[t]
    \centering
    \includegraphics[width=\textwidth]{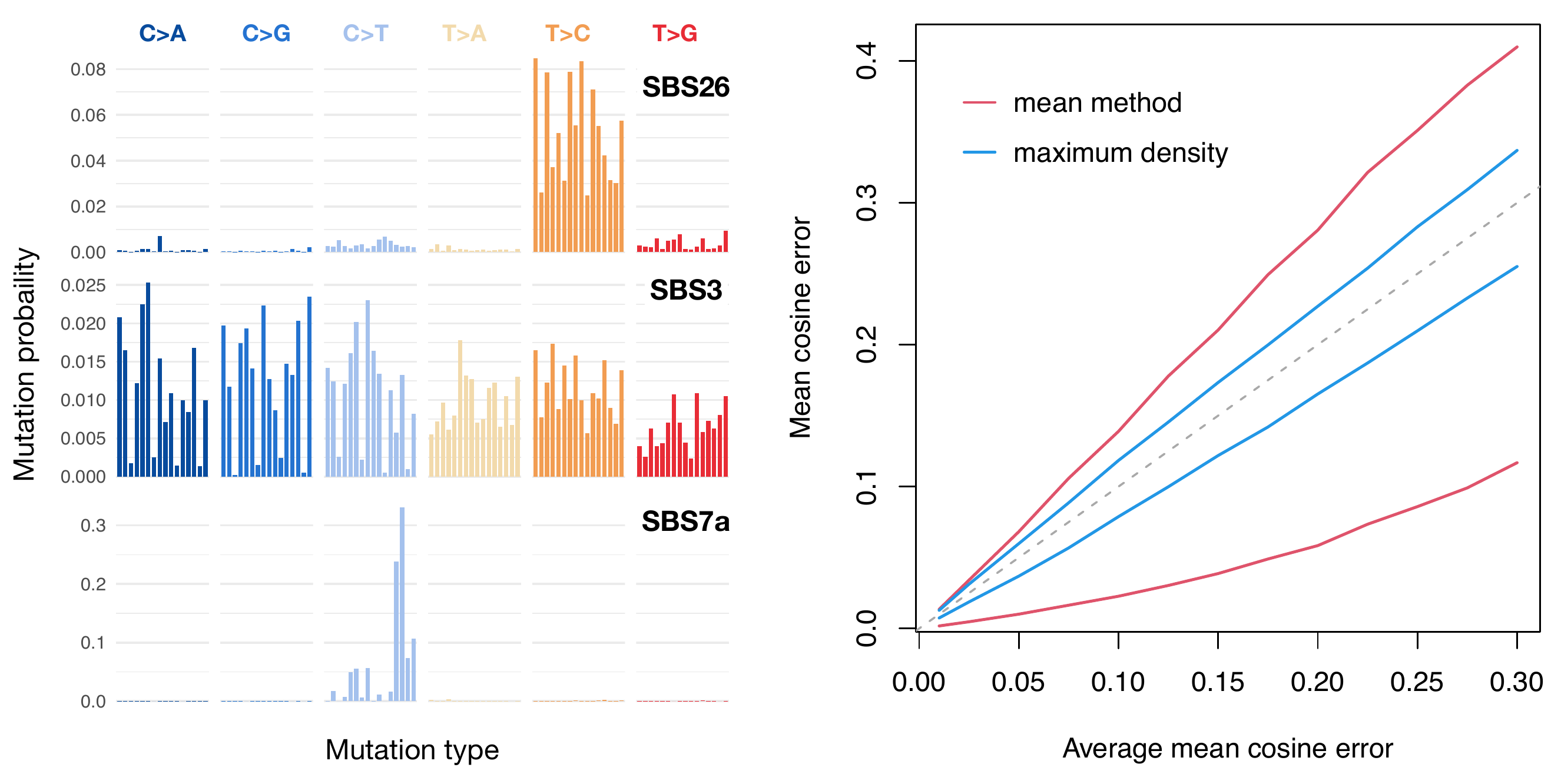}
    \caption{(\emph{Left}) Three examples of single-base substitution signatures in COSMIC v3.4: SBS26 is associated with defective DNA mismatch repair, SBS3 with homologous recombination deficiencies, and SBS7a with UV light exposure. (\emph{Right}) 25th and 75th percentiles of the empirical distribution of mean cosine errors across all 86 COSMIC v3.4 signatures.
    For each signature $c$, we compute the mean cosine error $\E(\mathrm{CosErr}(X,\E(X)))$ when $X$ follows a Dirichlet specified by the mean method as $\alpha$ varies (red lines) or maximum density method as $\kappa$ varies (blue lines) with target location $c$. The lines show the percentiles of the empirical distribution of these values across COSMIC signatures $c$, as a function of the average of this empirical distribution at given value of $\alpha$ or $\kappa$.}
    \label{fig:mutational-signatures}
\end{figure}

We propose to instead use our maximum density method to specify Dirichlet distributions for representing variability around mutational signatures.  
Specifically, for a given COSMIC signature $c$ and a desired mean cosine error $\kappa$, we aim to maximize $\mathrm{Dirichlet}(c\mid a_1,\ldots,a_K)$ subject to the constraint that $\E(\mathrm{CosErr}(X,\E(X))) = \kappa$, where $X\sim\mathrm{Dirichlet}(a_1,\ldots,a_K)$.  Since the mean cosine error $\E(\mathrm{CosErr}(X,\E(X)))$ is not mathematically tractable, we use the expected value of a second-order Taylor approximation to the cosine error; see \cref{sec:taylor-cosine} for details.
We then use \cref{alg:dirichlet} to maximize $\mathrm{Dirichlet}(c\mid a_1,\ldots,a_K)$ subject to constraining this approximation to the mean cosine error; see \cref{sec:derivation-of-algorithm} for the precise formulas we use for the constraint function $h(a)$ and its Jacobian matrix $J$.
We find that this algorithm converges reliably in all of the settings we have tried.

\cref{fig:mutational-signatures} (right) shows that our maximum density method provides much better control over the mean cosine error.  
More precisely, the distribution of mean cosine errors across COSMIC signatures is much tighter, meaning that mean cosine error is being more effectively controlled.
This is shown by having a smaller gap between 25th and 75th percentiles of the distribution of mean cosine errors across COSMIC signatures.
This demonstrates that the maximum density method provides better control over the scale of Dirichlet distributions in this real world application.
This example also illustrates the flexibility of the maximum density method, since one can choose a measure of scale that is relevant for the application at hand.

%\begin{equation}\label{eq:Poisson}
%X_{i j} \sim \mathrm{Poisson}\bigg(\sum_{k = 1}^K r_{i k}\theta_{k j}\bigg),
%\end{equation}
%The following priors are imposed over the model. 
%\begin{equation}\label{eq:priors}
%r_k \sim \mathrm{Dir}(\alpha_{1k}, \ldots, \alpha_{Ik}), \quad \theta_{kj}\mid \mu_k \sim \mathrm{Ga}(a, a/\mu_k), \quad \mu_k\sim \mathrm{InvGa}(a J +1, \varepsilon a J).
%\end{equation}

%The crucial task is to appropriately tune the parameters $\alpha_{1k}, \ldots, \alpha_{Ik}$ to ensure appropriate solutions. 

%\subsection{Simulation setting}

%\aletodo{[If we are satisfied with the application, might potentially move this section to the supplement. If this is the case, no need to have subsections]}

%\subsection{Application to pancreatic cancer}

%\aletodo{[Undecided on whether to do 21 breast cancer, or try the pancreatic PDAC. Last time, it only detected two signatures, while this approach might do well with three! ]}

%%%%%%%%%%%%%%%%%%%%%%%%%%%%%%%%%%%%%%%%%%%%%
% Discussion
%%%%%%%%%%%%%%%%%%%%%%%%%%%%%%%%%%%%%%%%%%%%%
\section{Discussion}
\label{sec:discussion}

The usual way of specifying Beta and Dirichlet distributions, by setting the mean equal to the target location, is prone to exhibit pathological behavior near the boundary. 
This underappreciated problem may lead to poorly performing MCMC algorithms or unintentionally strong priors that induce significant bias.
We introduce a novel approach that provides better control over the location and scale of these distributions.

This issue---and our proposed solution---may be relevant in a wide range of applications in which Dirichlet distributions play a central role, including
mutational signatures analysis \citep{zito2024compressivebayesiannonnegativematrix},
microbiome analysis of compositional count matrices \citep{chen2013variable}, species abundance data \citep{Bersson_Hoff_2024},
document analysis with topic models \citep{blei2003latent},
population structure analysis with admixture models \citep{pritchard2000inference},
and applications of mixture models \citep{Miller_Harrison_2018}. 
% and variable selection \citep{Bhattacharya_2015, Zhang_Bondell_2018}.

Another application of particular interest is variable selection with global-local shrinkage priors. 
Dirichlet--Laplace priors \citep{Bhattacharya_2015, Zhang_Bondell_2018} are used to model the local shrinkage parameters of the mean of each observation using a Dirichlet distribution, while the global variance is assigned a gamma prior. With current methods, inference is only feasible using Gibbs sampling for very specific choices of the gamma hyperparameters. Our maximum density method may facilitate the design of efficient proposals for the local parameters when they are close to zero.

Finally, the maximum density method could be applied to other distributions, beyond the Beta and Dirichlet.
The technique of maximizing the density at a target location, subject to a scale constraint, makes sense for many families of distributions.

\section*{Acknowledgments}
C.X.\ was supported by NIH Training Grant T32GM135117 and NSF Graduate Research Fellowship DGE-2140743. 
J.W.M.\ and A.Z.\ were supported in part by the National Cancer Institute of the National Institutes of Health under award number R01CA240299. 
The content is solely the responsibility of the authors and does not necessarily represent the official views of the National Institutes of Health.

\section*{Supplementary Material}
The Supplementary Material includes further results and analyses. 
Code implementing the method and examples in the paper is publicly available via at \url{https://github.com/casxue/ImprovedDirichlet}.

\bibliographystyle{ba}
\bibliography{references}

%%%%%%%%%%%%%%%%%%%%%%%%%%%%%%%%%%%%%%%%%%%%%%%%%%%%%%%%%%%%%%%%%%%%%%%%%%%%%%%%
% Supplementary material

% \clearpage
\newpage
\setcounter{page}{1}
\setcounter{section}{0}
\setcounter{table}{0}
\setcounter{figure}{0}
\setcounter{theorem}{0}
\setcounter{algorithm}{0}
\renewcommand{\theHsection}{SIsection.\arabic{section}}
\renewcommand{\theHtable}{SItable.\arabic{table}}
\renewcommand{\theHfigure}{SIfigure.\arabic{figure}}
\renewcommand{\theHtheorem}{SItheorem.\arabic{theorem}}
\renewcommand{\theHalgorithm}{SIalgorithm.\arabic{algorithm}}
\renewcommand{\thepage}{S\arabic{page}}  
\renewcommand{\thesection}{S\arabic{section}}   
\renewcommand{\thetable}{S\arabic{table}}  
\renewcommand{\thefigure}{S\arabic{figure}}
\renewcommand{\thetheorem}{S\arabic{theorem}}   
\renewcommand{\thealgorithm}{S\arabic{algorithm}}

\makeatletter
\def\ps@myheadings{%
  \def\@oddhead{}%
  \def\@evenhead{}%
  \def\@oddfoot{\hfil \thepage \hfil}%
  \def\@evenfoot{\hfil \thepage \hfil}}
\makeatother

\pagestyle{myheadings}

% \begin{bibunit}
% \input{supplement.tex}
% \putbib
% \end{bibunit}

\begin{center}
{\Large Supplementary Material for \\``Improved control of Dirichlet location and scale near the boundary''}
\end{center}

\section{Proofs}
\label{sec:proofs}

This section provides the proofs of the results in the paper.

\begin{proof}[\bf Proof of \cref{thm:optimization}]
Let $\Theta = (0,\infty)^2\subseteq\mathbb{R}^2$.
Define $f:\Theta\to\mathbb{R}$ by 
$f(a,b) = -\log\mathrm{Beta}(c\mid a,b)$ and define $S = \{(a,b)\in\Theta : V(a,b) = v\}$.
In terms of $f$ and $S$, the optimization problem in \cref{thm:optimization} is 
\begin{align}\label{eq:opt-problem}
    \argmin_{a,b>0} f(a,b) \text{ subject to } (a,b)\in S.
\end{align}
First, $S$ is nonempty since by \cref{thm:beta-existence} there exists a Beta distribution with variance $v$ whenever $0 < v < 1/4$.
Furthermore, $S$ is bounded, since
\begin{align*}
    \min\{1/a,1/b\} > \frac{1}{a+b+1} \geq \frac{a}{a+b}\,\frac{b}{a+b}\,\frac{1}{a+b+1} = V(a,b) = v > 0
\end{align*}
for all $(a,b)\in S$, and therefore $\max\{a,b\} < 1/v$ for all $(a,b)\in S$.  
% Hence, $S \subseteq (0,1/v]^2$.
Now, for all $a,b\in(0,1]$, we have $\mathrm{B}(a,b) \geq (a + b - a b)/(a b) = 1/a + 1/b - 1$, by \citet{ivady2016}; also see \citet{zhao2023lower}.
Therefore, for $a,b\in(0,1)$,
\begin{align}\label{eq:obj-inequality}
    f(a,b) &= -(a-1)\log(c) - (b-1)\log(1-c) + \log\mathrm{B}(a,b) \notag \\
    &\geq \log(c) + \log(1-c) + \log(1/a + 1/b - 1).
\end{align}
Let $(a',b')$ be any point in $S$.
By \cref{eq:obj-inequality}, we can choose $\varepsilon\in(0,1)$ small enough that $f(a,b) > f(a',b') + 1$ for all $(a,b)$ such that $a<\varepsilon$ or $b<\epsilon$.
Define $A = [\varepsilon,\infty)^2$.
It follows that $(a',b')\in A\cap S$ and $f(a',b') < \inf\{f(a,b) : (a,b)\in\Theta\setminus A\}$. Thus, no point of $\Theta\setminus A$ can be a solution, so minimizing over $S$ is equivalent to minimizing over $A\cap S$.  To complete the proof, we just need to show that the minimum of $f$ over $A\cap S$ is attained at some point; any such point will be a solution to \cref{eq:opt-problem}.

Let $V|_A$ denote the restriction of $V$ to $A$.  Then $V|_A$ is a continuous function, so the pre-image of $\{v\}$ is closed, that is, $\{(a,b)\in A : V(a,b) = v\} = A \cap S$ is a closed subset of $A$, and hence, a closed subset of $\mathbb{R}^2$.
Since $A\cap S$ is also bounded (as a subset of the bounded set $S$), it is a compact set.
Therefore, the continuous function $f$ attains its minimum on $A\cap S$.
\end{proof}

\begin{theorem}
    \label{thm:optimization-alpha}
    For all $c\in(0,1)$ and $\alpha\in(0,\infty)$, there exists a finite solution to:
    $$ \argmax_{a,b>0} \mathrm{Beta}(c\mid a,b) \text{ subject to } a + b = \alpha. $$
\end{theorem}
\begin{proof}
The proof is essentially the same as the proof of \cref{thm:optimization}, except that $S = \{(a,b)\in\Theta : a + b = \alpha\}$. All that is needed is to recognize that $S$ is a nonempty, bounded set, and the rest of the proof is the same.
\end{proof}

\begin{proof}[\bf Proof of \cref{thm:beta-existence}]
Reparametrize the $\mathrm{Beta}(a,b)$ distribution in terms of the mean $u = a/(a+b)$ and concentration $\alpha = a+b$, so that $a = \alpha u$ and $b = \alpha(1-u)$.
The variance $v$ of $\mathrm{Beta}(a,b)$ is related to $u$ and $\alpha$ via
\begin{align}
\label{eq:v}
v = \frac{a b}{(a + b)^2 (a + b + 1)} = \frac{u(1-u)}{\alpha + 1}.
\end{align}
Thus, solving for $\alpha$ as a function of $u$ and $v$ yields that $\alpha = u(1-u)/v - 1$.
For any $v>0$, define 
$$ f_v(u) = \frac{u(1-u)}{v} - 1. $$
Since $\alpha$ must be positive to be the concentration of a Beta distribution, only a subset of $u$ and $v$ values can be the mean and variance of a Beta distribution.
More precisely, for any $u\in(0,1)$ and $v>0$, there is a Beta distribution with mean $u$ and variance $v$ if and only if $f_v(u) > 0$.

Note that $u \mapsto u(1-u)$ is a quadratic function with a maximum of $u(1-u) = 1/4$ at $u = 1/2$.
Along with Equation \ref{eq:v}, this implies that the variance must satisfy $0 < v < 1/4$ since $\alpha > 0$.
To find the feasible range of $u$ values for any given $v$, we set
$$ 0 = f_v(u) = \frac{u(1-u)}{v} - 1 = - u^2/v + u/v - 1 $$
and solve using the quadratic formula to find
\begin{align}
u_* &= \frac{-1/v + \sqrt{1/v^2 - 4/v}}{-2/v} = 1/2 - (1/2)\sqrt{1 - 4 v} \label{eq:u_lower} \\
u^* &= \frac{-1/v - \sqrt{1/v^2 - 4/v}}{-2/v} = 1/2 + (1/2)\sqrt{1 - 4 v}. \label{eq:u_upper}
\end{align}
Thus, for any $v \in (0,1/4)$, we have $f_v(u) > 0$ if and only if $u_* < u < u^*$, where $u_*$ and $u^*$ given by Equations \ref{eq:u_lower} and \ref{eq:u_upper}.
Therefore, there is a Beta distribution with mean $u$ and variance $v$ if and only if $0 < v < 1/4$ and $|u - 1/2| < (1/2)\sqrt{1 - 4 v}$.
\end{proof}

\section{Distance between logit-transformed values}
\label{sec:logit-derivation}

In this section, we derive the density of the distance between a Beta random variable and a given point, after logit transformation.
Furthermore, we describe how to compute this density in a numerically stable way.

Fix $a>0$, $b>0$, and $c\in(0,1)$.  Let $X\sim\mathrm{Beta}(a,b)$ and define
\begin{align*}
    Y = \big\vert\mathrm{logit}(X) - \mathrm{logit}(c)\big\vert,
\end{align*}
where $\mathrm{logit}(x) = \log(x/(1-x))$ for $x\in(0,1)$. 
Let $f_X(x)$ denote the probability density function of $X$, namely, \begin{align*}
    f_X(x) = \frac{1}{\mathrm{B}(a,b)} x^{a-1}(1-x)^{b-1} \mathds{1}(0<x<1).
\end{align*}
In \cref{sec:density-derivation}, we show that the probability density function of $Y$ is
\begin{align}
\label{eq:logit-density-y}
    f_Y(y) = f_X\bigg(\frac{1}{1 + e^{y-\ell}}\bigg) \frac{e^{y-\ell}}{(1 + e^{y-\ell})^2} + f_X\bigg(\frac{1}{1 + e^{-y-\ell}}\bigg) \frac{e^{-y-\ell}}{(1 + e^{-y-\ell})^2}
\end{align}
for $y\in(0,\infty)$, and $f_Y(y) = 0$ otherwise.

\subsection{Numerically stable computation of the density of $Y$}

Computing $f_Y(y)$ requires careful handling of the exponentials in order to avoid numerical underflow or overflow issues. The first point to note is that one should work with logarithms rather than the values themselves. Specifically, compute the log of each term in \cref{eq:logit-density-y}, and combine them using the ``logsumexp'' trick to obtain $\log f_Y(y)$, that is, $\log(\exp(r) + \exp(s)) = \log(\exp(r-m) + \exp(s-m)) + m$ where $m = \max\{r,s\}$.
The next point is that computing $\log(1 + e^x)$ is prone to numerical issues: When $x \ll 0$, it rounds off to $0$, whereas when $x \gg 0$, it overflows to $\infty$.  This can be fixed by (i) using a $\texttt{log1p}$ function, which computes $\log(1+t)$ in a numerically accurate way for small $t$, along with (ii) a conditional to avoid cases where $e^x$ would overflow, for instance:
\begin{align*}
\log(1 + e^x) \approx \begin{cases}
        \texttt{log1p}(\texttt{exp}(x)) & \text{if } x < 0 \\
        \texttt{log1p}(\texttt{exp}(-x))+x & \text{if } x \geq 0.
    \end{cases}  
\end{align*}

\subsection{Derivation of the density of $Y$}
\label{sec:density-derivation}

To derive \cref{eq:logit-density-y}, we use Jacobi's formula for transformation of continuous random random variables to derive the probability density function of $Y$, as follows.
Let $A_0 = \{c\}$, $A_1 = (0,c)$, and $A_2 = (c,1)$, noting that  $\{A_0,A_1,A_2\}$ is a partition of $\mathcal{X} := \{x : f_X(x) > 0\} = (0,1)$ such that $P(X \in A_0) = 0$.
Define $g(x) = |\mathrm{logit}(x) - \mathrm{logit}(c)|$, and observe that 
$g$ is strictly monotone on each of $A_1$ and $A_2$, separately.
Furthermore, letting $g_1$ and $g_2$ be the restrictions of $g$ to $A_1$ and $A_2$, respectively, it holds that the inverse $g_i^{-1}$ exists and is continuously differentiable on $g(A_i) = \{g(x) : x\in A_i\}$, for $i=1,2$.
Thus, by Jacobi's transformation formula \citep{casella2024statistical,jacod2012probability}, the probability density of $Y$ is
\begin{align}
\label{eq:logit-density}
    f_Y(y) = \sum_{i=1}^2 f_X(g_i^{-1}(y))\Big\vert\frac{d}{d y} g_i^{-1}(y)\Big\vert \mathds{1}(y\in g(A_i)).
\end{align}
To use this formula, we just need to derive $g_i^{-1}$ and its derivative for $i=1,2$.
To this end, observe that 
\begin{align*}
    g_1(x) = \log\Big(\frac{c}{1-c}\Big) - \log\Big(\frac{x}{1-x}\Big) & ~~~~~~~~ & x\in (0,c), \\
    g_2(x) = \log\Big(\frac{x}{1-x}\Big) - \log\Big(\frac{c}{1-c}\Big) & ~~~~~~~~ & x\in (c,1).
\end{align*}
Define $\ell = \log(c / (1-c))$ to simplify the notation.  Solving for the inverses, we obtain
\begin{align}
    g_1^{-1}(y) = \frac{1}{1 + e^{y - \ell}} & ~~~~~~~~ & y\in(0,\infty), \label{eq:logit-inverses1} \\
    g_2^{-1}(y) = \frac{1}{1 + e^{-y - \ell}} & ~~~~~~~~ & y\in(0,\infty), \label{eq:logit-inverses2}
\end{align}
and differentiating, we have
\begin{align}
    \frac{d}{d y} g_1^{-1}(y) = \frac{-e^{y - \ell}}{(1 + e^{y - \ell})^2} & ~~~~~~~~ & y\in(0,\infty), \label{eq:logit-derivatives1}\\
    \frac{d}{d y}g_2^{-1}(y) = \frac{e^{-y-\ell}}{(1 + e^{-y - \ell})^2} & ~~~~~~~~ & y\in(0,\infty). \label{eq:logit-derivatives2}
\end{align}
Plugging \cref{eq:logit-inverses1,eq:logit-inverses2,eq:logit-derivatives1,eq:logit-derivatives2} into \cref{eq:logit-density} yields 
\begin{align}
    f_Y(y) = f_X\bigg(\frac{1}{1 + e^{y-\ell}}\bigg) \frac{e^{y-\ell}}{(1 + e^{y-\ell})^2} + f_X\bigg(\frac{1}{1 + e^{-y-\ell}}\bigg) \frac{e^{-y-\ell}}{(1 + e^{-y-\ell})^2}
\end{align}
for $y\in(0,\infty)$, and $f_Y(y) = 0$ otherwise.

\section{Derivation of the optimization algorithm}
\label{sec:derivation-of-algorithm}

In this section, we derive \cref{alg:dirichlet}
based on Newton's method with equality constraints \citep{boyd2004convex}.
This technique provides a second-order optimization algorithm for solving problems of the form
\begin{align*}
    \arg\min_{x} f(x) \text{ subject to } h(x) = 0,
\end{align*}
where $f$ is a strictly convex, twice-differentiable function and $h$ is differentiable. 
% \jefftodo{Might need $h$ to be linear to have a guarantee here.}
Let $g(x) = \nabla f(x) = (\partial f/\partial x_i)$ denote the gradient of $f$, let $H(x) = \nabla^2 f(x) = (\partial^2 f/\partial x_i\partial x_j)$ denote the Hessian matrix of $f$, and let $J(x) = (\partial h_i / \partial x_j)$ denote the Jacobian matrix of $h$.
After initializing $x$ to an appropriate value, each iteration proceeds by updating 
\begin{align}\label{eq:newton-update}
    x \gets x + \delta
\end{align}
where the vector $\delta$ is defined by solving the linear system 
\begin{align}\label{eq:linear-system}
    \begin{bmatrix}
        H(x) & J(x)^\mathtt{T} \\
        J(x) & 0
    \end{bmatrix}
    \begin{bmatrix}
        \delta \\ \lambda
    \end{bmatrix}
    = 
    \begin{bmatrix}
        -g(x) \\ -h(x)
    \end{bmatrix}.
\end{align}
Here, $\lambda$ is a vector of multipliers that we will not use in our algorithm.

\subsection{Applying Newton's method to the maximum density method}

To apply this to implement the maximum density method as in \cref{eq:max_density_dirichlet}, we define $x = a = (a_1,\ldots,a_K)$ and
\begin{align*}
    f(a) = -\log \mathrm{Dirichlet}(c\mid a) 
    = \sum_{i=1}^K \log\Gamma(a_i) - \log\Gamma\big({\textstyle\sum_{i=1}^K} a_i\big) - \sum_{i=1}^K (a_i - 1)\log(c_i).
\end{align*}
Then $f$ is smooth and strictly convex on $(0,\infty)^K$ by \citet[Proposition 19]{miller2014inconsistency}.
For the constraint in the Dirichlet case, we consider two options: (i) fixed concentration $\alpha$,  and
(ii) fixed mean cosine error $\kappa$.
For the special case of Beta distributions, see \cref{sec:derivation-of-algorithm-beta}.
We show how to handle each of these constraints below.

The gradient and Hessian of $f(a)$ are obtained by differentiating:
\begin{align*}
    \frac{\partial f}{\partial a_i}(a) &= \psi(a_i) - \psi\big({\textstyle\sum_i} a_i\big) - \log(c_i), \\
    \frac{\partial^2 f}{\partial a_i \partial a_j}(a) &= \psi'(a_i)\mathds{1}(i=j) - \psi'\big({\textstyle\sum_i} a_i\big).
\end{align*}
Thus, the functions $g$ and $H$ take the following forms, where $s = \sum_i a_i$:
\begin{align*}
    g(a) &= \nabla f(a) = \big[\psi(a_i) - \psi(s) - \log(c_i)\big]\in\mathbb{R}^{K\times 1} \\
    H(a) &= \nabla^2 f(a) = \big[\psi'(a_i)\mathds{1}(i=j) - \psi'(s)\big]\in\mathbb{R}^{K\times K}
\end{align*}
where $\psi(x) = (d/d x) \log\Gamma(x)$ is the digamma function and $\psi'(x) = (d/d x) \psi(x)$ is the trigamma function.

\subsubsection{Constraining the concentration parameter}
First, we consider constraining the  concentration parameter $\alpha$,
in which case we define
\begin{align}\label{eq:fixed-concentration-constraint}
    h(a) = \frac{\sum_{i=1}^K a_i}{\alpha} - 1.
\end{align}
Note that $h(a) = 0$ if and only if $\alpha = \sum_{i=1}^K a_i$.
In this case, $h$ is a linear function of $a$, which is conducive for convergence since we are optimizing a strictly convex function over a convex set.
The rationale for defining $h$ using the ratio $(\sum_i a_i) / \alpha$ rather than the difference $(\sum_i a_i) - \alpha$ is so that the convergence tolerance can be relatively invariant to the magnitude of $\alpha$; when using the ratio, the number of significant digits is what matters.
The Jacobian matrix is obtained by differentiating \cref{eq:fixed-concentration-constraint},
\begin{align*}
    \frac{\partial h}{\partial a_i} = \frac{\partial}{\partial a_i} \Big(\frac{\sum_{i=1}^K a_i}{\alpha} - 1\Big) = \frac{1}{\alpha},
\end{align*}
and thus, 
\begin{align*}
    J(a) = \begin{bmatrix}
        1/\alpha & \cdots & 1/\alpha
    \end{bmatrix}.
\end{align*}

\subsubsection{Constraining the mean cosine error}
For the case of mean cosine error, we use the following approximation to the mean cosine error between $X\sim\mathrm{Dirichlet}(a_1,\ldots,a_K)$ and $\mathbb{E}(X)$,
\begin{align}
    \mathbb{E}(\mathrm{CosErr}(X, \mathbb{E}(X))) \approx \frac{\sum_i a_i}{2(1+\sum_i a_i)(\sum_i a_i^2)}\bigg({\textstyle\sum_i} a_i - \frac{\sum_i a_i^3}{\sum_i a_i^2}\bigg)
\end{align}
where each sum is over $i=1,\ldots,K$; see \cref{sec:taylor-cosine} for the derivation.
% and $u = \mathbb{E}(X) = (a_1/\alpha,\ldots,a_K/\alpha)$ with $\alpha = \sum_{i=1}^K$.
For the constraint function $h$, we work with the logarithm, defining
\begin{align*}
    h(a) = & -\log(2) + \log({\textstyle\sum_i} a_i) - \log(1 + {\textstyle\sum_i} a_i) - \log({\textstyle\sum_i} a_i^2) \\
   & + \log\Big({\textstyle\sum_i} a_i - \frac{\sum_i a_i^3}{\sum_i a_i^2}\Big) - \log(\kappa)
\end{align*}
where $\kappa > 0$ is the desired mean cosine error.
The rationale for defining $h$ using the logarithm is so that the derivatives take mathematically simple forms.  Additionally, this expresses the constraint in terms of  the ratio rather than the difference, which is advantageous for the same reasons discussed in the concentration parameter case.
Differentiating, we have
\begin{align*}
\frac{\partial h}{\partial a_j} = 
\frac{1}{\sum_i a_i} - \frac{1}{1 + \sum_i a_i} - \frac{2 a_j}{\sum_i a_i^2} + \frac{1 - \big(3 a_j^2 (\sum_i a_i^2) - 2 a_j (\sum_i a_i^3)\big)/(\sum_i a_i^2)^2}{(\sum_i a_i) - (\sum_i a_i^3)/(\sum_i a_i^2)}.
% \frac{\partial h}{\partial a_j} = \frac{1}{\sum_i a_i} - \frac{1}{1 + \sum_i a_i} - \frac{2 a_j}{\sum_i a_i^2} + \frac{1}{\sum_i a_i^2} \frac{(\sum_i a_i^2)^2 - \big(3 a_j^2 (\sum_i a_i^3) - 2 a_j (\sum_i a_i^3)\big)}{(\sum_i a_i)(\sum_i a_i^2) - \sum_i a_i^3}.
\end{align*}
Defining $s_k = \sum_i a_i^k$ for $k\in\{1,2,3\}$, the formulas can be expressed more compactly as
\begin{align*}
    h(a) &= -\log(2) + \log(s_1) - \log(1+s_1) - \log(s_2) + \log(s_1 - s_3/s_2) - \log(\kappa)\\
    \frac{\partial h}{\partial a_j} &= \frac{1}{s_1} - \frac{1}{1+s_1} - \frac{2 a_j}{s_2} + \frac{1 - (3 a_j^2 s_2 - 2 a_j s_3)/s_2^2}{s_1 - s_3/s_2}.
\end{align*}
Even though $h$ is a nonlinear function of $a$, we find that the algorithm still successfully converges once some implementation details have been addressed; we discuss this next.

\subsection{Implementation details}
\label{sec:implementation-details}

To implement the basic version of the algorithm described above, at each iteration we would compute the expressions for $g(a)$, $H(a)$, and $J(a)$, plug them into \cref{eq:linear-system}, use a linear solver to obtain $\delta$, and update $a$ as in \cref{eq:newton-update}.
In practice, we modify the algorithm to improve its numerical stability. First, we modify \cref{eq:newton-update} to use a step size of $\rho\in(0,1)$, so that the update is
\begin{align}\label{eq:newton-update-stepsize}
    a \gets a + \rho\,\delta.
\end{align}
Second, we enforce the boundary constraint that $a_1,\ldots,a_K$ must be positive numbers as follows: for each $i=1,\ldots,K$, if $a_i$ would be less than or equal to zero after the update in \cref{eq:newton-update-stepsize}, then we instead update $a_i \gets a_i/2$ for that $i$. This maintains positivity, but still moves $a_i$ in the direction of the full Newton step.

We begin with an initial step size of $\rho = 1/2$, and we initialize the algorithm at $a = 10(c+1)/2$, that is, $a_i = 10(c_i+1)/2$ for $i=1,\ldots,K$.
% The motivation for this choice of initialization is that $(c+1)/2$ is in the direction of the target location $c$ but is regularized towards $(1,\ldots,1)$; while multiplying by $10$ makes it 
As a stopping criterion, we halt the algorithm after either (i) a maximum number of iterations $\texttt{maxiter}$ has been reached (default: $\texttt{maxiter} = 100$), or (ii) $|h(a)| + \sum_{i=1}^K |a_i/a_i' - 1| < \texttt{tol}$, where $a_1',\ldots,a_K'$ are the values of $a_1,\ldots,a_K$ from the previous iteration and $\texttt{tol} > 0$ is a convergence tolerance (default: $\texttt{tol} = 10^{-8}$).

If the algorithm reaches the maximum number of iterations without converging to within the specified tolerance, we restart the algorithm with $\rho \gets \rho/5$ and $\texttt{maxiter} \gets 5\,\texttt{maxiter}$. If the algorithm still fails after 5 adaptive restarts of this form, we stop and report that the procedure has failed.
We find that this adaptive procedure reliably yields successful convergence for a wide range of values.

\subsection{Algorithm in the special case of Beta distributions}
\label{sec:derivation-of-algorithm-beta}

For concreteness, \cref{alg:beta} provides a version of \cref{alg:dirichlet} that is specialized to Beta distributions. To constrain the concentration to some value $\alpha$, lines 8--10 of \cref{alg:beta} become
$h \gets (a+b)/\alpha - 1$, $J_1 \gets 1/\alpha$, and $J_2 \gets 1/\alpha$.
Meanwhile, to constrain the variance to some value $v\in(0,1/4)$, lines 8--10 become
\begin{align*}
& h \gets \log(a) + \log(b) - 2\log(a+b) - \log(a+b+1) - \log(v) \\
& J_1 \gets 1/a - 2/(a+b) - 1/(a+b+1) \\
& J_2 \gets 1/b - 2/(a+b) - 1/(a+b+1).
\end{align*}

\begin{algorithm}
\caption{~~Maximum density method for Beta distributions}\label{alg:beta}
% \raggedright
\begin{algorithmic}[1]
\Require{target location $c\in(0,1)$, step size $\rho\in(0,1]$, maximum number of iterations $\texttt{maxiter}>0$, and convergence tolerance $\texttt{tol}>0$.}
\State $a \gets c$ and $b \gets 1-c$ \Comment{Initialization}
\For{$i = 1,\ldots,\texttt{maxiter}$}
    \State $g_1 \gets -\log(c) - \psi(a+b) + \psi(a)$ \Comment{Compute gradient}
    \State $g_2 \gets -\log(1-c) - \psi(a+b) + \psi(b)$
    \State $H_{1 1} \gets -\psi'(a+b) + \psi'(a)$ \Comment{Compute Hessian matrix}
    \State $H_{2 2} \gets -\psi'(a+b) + \psi'(b)$
    \State $H_{1 2} \gets -\psi'(a+b)$
    \State $h \gets h(a,b)$ \Comment{Constraint}
    \State $J_1 \gets \partial h/\partial a$ \Comment{Jacobian of constraint}
    \State $J_2 \gets \partial h/\partial b$
    \State Solve for $\delta_1$ and $\delta_2$ in the linear system: \Comment{Compute Newton step}
        \begin{align*}
        \begin{bmatrix}
            H_{1 1} & H_{1 2} & J_1 \\
            H_{1 2} & H_{2 2} & J_2 \\
            J_1     & J_2     & 0
        \end{bmatrix}
        \begin{bmatrix}
            \delta_1 \\ \delta_2 \\ \lambda
        \end{bmatrix}
        =
        \begin{bmatrix}
            -g_1 \\ -g_2 \\ -h
        \end{bmatrix}
        \end{align*}
    \State $a' \gets a$ and $b' \gets b$  \Comment{Store current values}
    \State $a \gets a' + \rho\,\delta_1$ and $b \gets b' + \rho\,\delta_2$ \Comment{Update values}
    \State \textbf{if} $a \leq 0$ \textbf{then} $a \gets a'/2$  \Comment{Enforce boundary constraints}
    \State \textbf{if} $b \leq 0$ \textbf{then} $b \gets b'/2$
    \If{$|h| + |a/a' - 1| + |b/b' - 1| < \texttt{tol}$} \Comment{Check for convergence}
    \State \textbf{output} $a$ and $b$ \Comment{Return output}
    \EndIf
\EndFor
\Ensure{Beta parameters $a$ and $b$.}
\end{algorithmic}
\end{algorithm}

\section{Additional details on the median method}
\label{sec:median}

As discussed in \cref{sec:method}, an alternative approach is to choose a Beta distribution with median equal to the target location $c$, and with either a given variance $v$ or a given concentration parameter $\alpha$. 
\cref{fig:median} shows the plots of CDFs and percentiles for the median method with $c\in\{0.001,0.2\}$, for a range of $\alpha$ values  (compare with \cref{fig:problem,fig:method}).
Like maximum density, the median method provides better control over the location and scale of Beta distributions than the mean method. The median method has the further advantage of directly centering the distribution at $c$ (in terms of the median) by construction.

% \jefftodo{(Maybe: Show that for all $m\in(0,1)$ and all $v\in(0,1/4)$, there exists a Beta distribution with median $m$ and variance $v$.)}

\begin{figure}
    \centering
    \includegraphics[width=\textwidth]{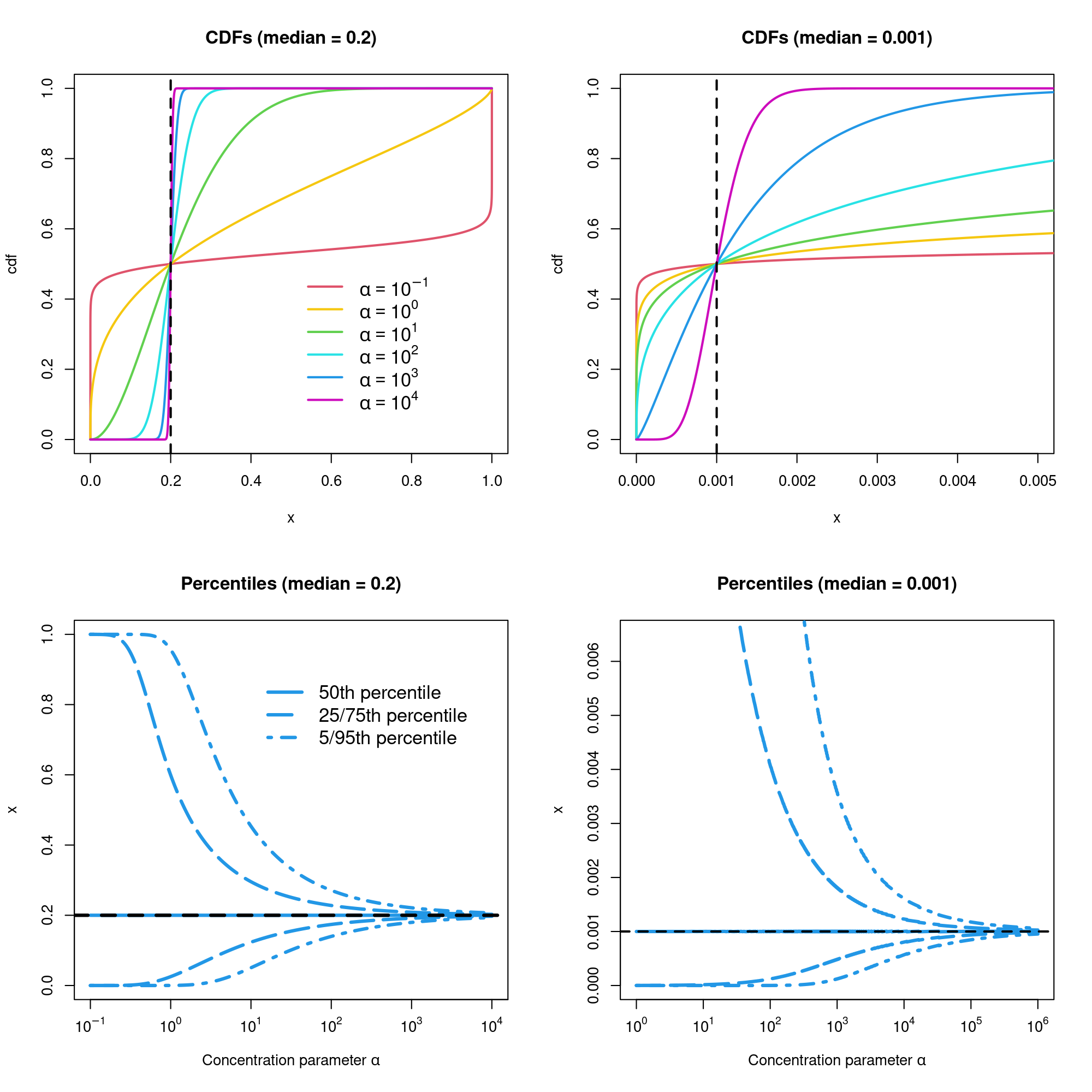}
    \caption{The median method performs similarly to our maximum density method in terms of providing improved control over the scale.  See the captions of \cref{fig:problem,fig:method} for description of the plots.}
    \label{fig:median}
\end{figure}

However, as shown in \cref{fig:logit-median}, the median method puts somewhat less mass near the target location $c$, compared to the maximum density method.
The objective function for median method is also more complicated to optimize, compared to maximizing the density as in our proposed method.
Furthermore, importantly, the median method is not as straightforward to extend to the general Dirichlet case.

\begin{figure}
    \centering
    \includegraphics[trim=0.5cm 0 1.25cm 0, clip, width=0.49\textwidth]{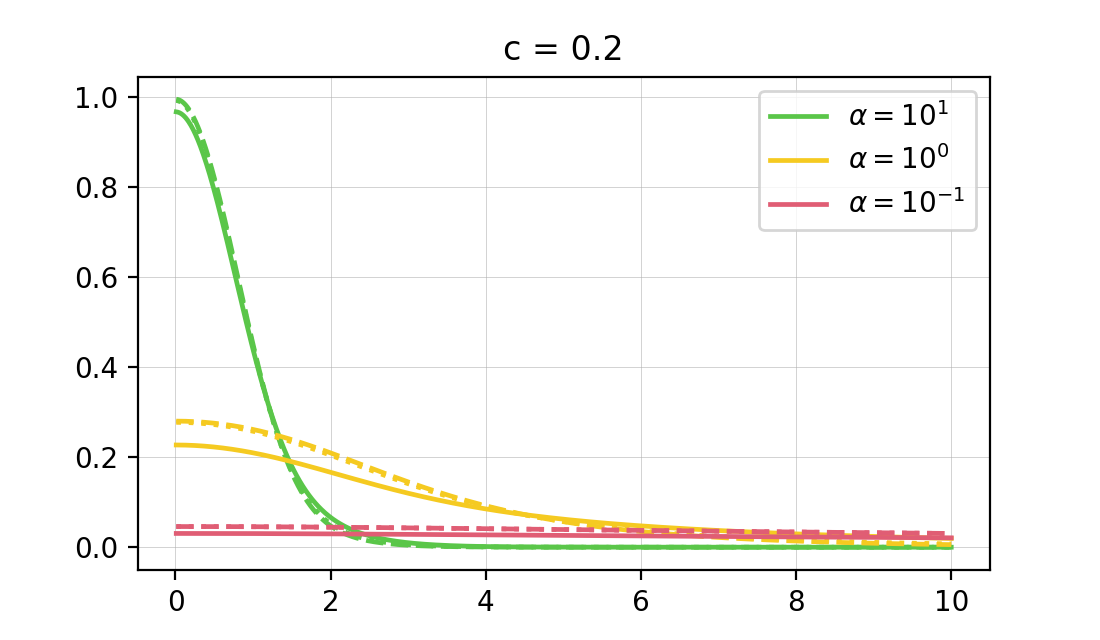}
    \includegraphics[trim=0.5cm 0 1.25cm 0, clip, width=0.49\textwidth]{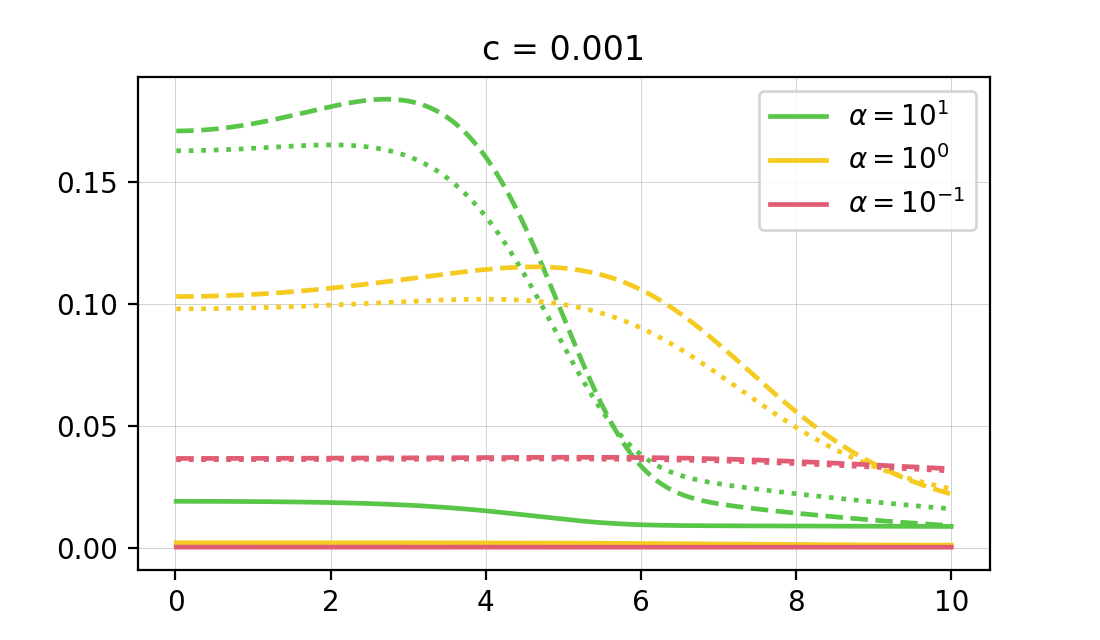}
    \caption{The maximum density method (dashed lines) puts more probability mass near the target location, compared to both the mean method (solid lines) and median method (dotted lines). See the caption of \cref{fig:logit} for description of the plots.  The concentration $\alpha$ is the same for all three methods.  The median and maximum density curves are nearly indistinguishable in the case of $c = 0.2$.}
    \label{fig:logit-median}
\end{figure}

On the Metropolis--Hastings example in \cref{sec:mh-example}, the median method performs similarly to the maximum density method as a technique for constructing MH proposal distributions. \cref{fig:autocorrelation-extended} shows the autocorrelation functions for the mean, median, and maximum density methods for a range of parameter settings. To facilitate comparison with the other methods, we consider proposals with median $c$ and (i) fixed concentration $\alpha = 5$, (ii) fixed variance $v = 0.1$, and (iii) adaptive variance, in other words, $\sigma = \min\{x, 1-x, \sqrt{0.1}\}$.

\begin{figure}
    \centering
    \includegraphics[width=\textwidth]{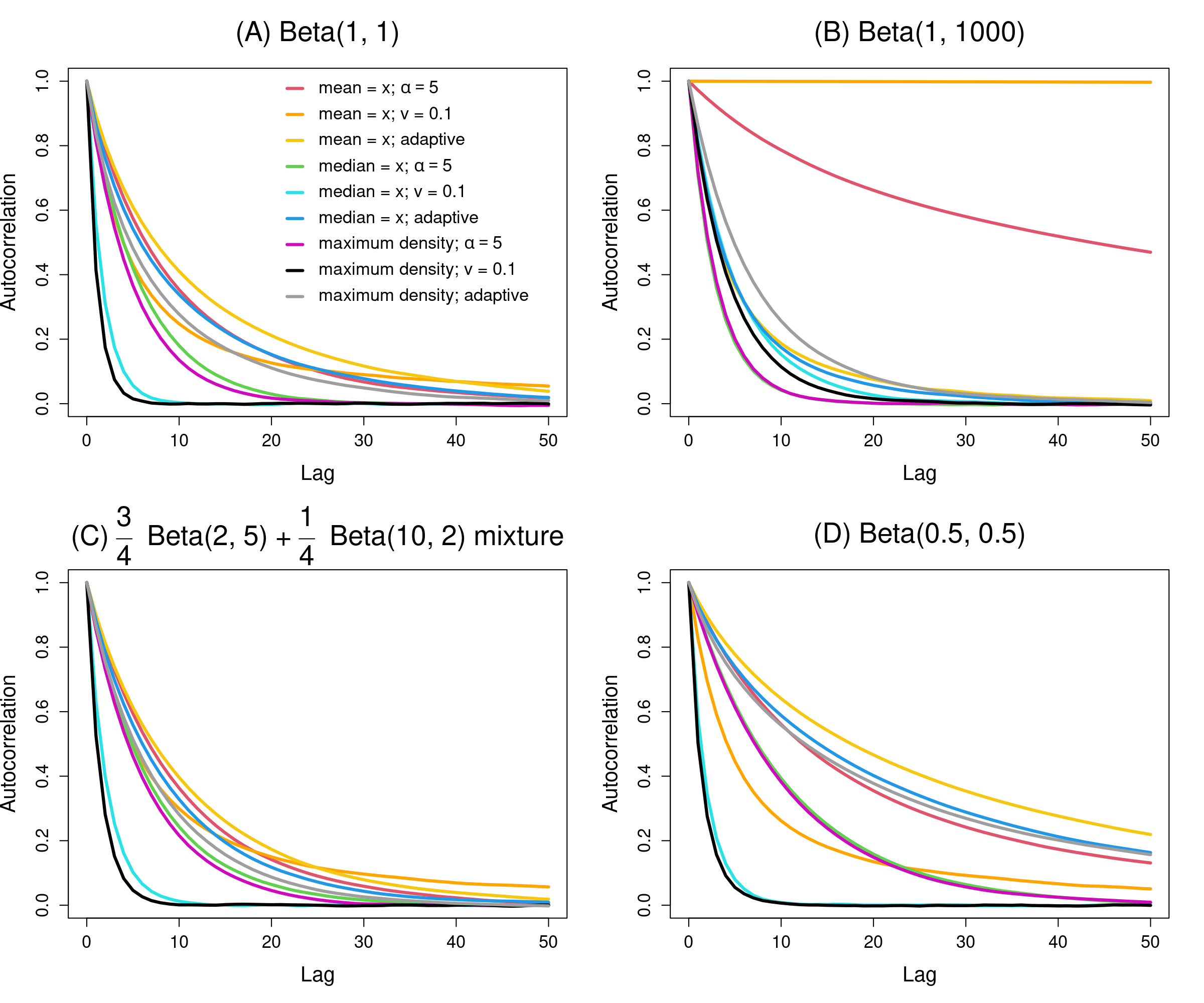}
    \caption{Autocorrelation functions for the mean, median, and maximum density methods of specifying a proposal distribution in the MH example in \cref{sec:mh-example}.  Compare with \cref{fig:autocorrelation}.}
    \label{fig:autocorrelation-extended}
\end{figure}

\section{Metropolis--Hastings example details}
\label{sec:mh-details}

This section provides more detail on the MCMC example in \cref{sec:mh-example}.
% \subsection{Details of the Metropolis--Hastings proposals and algorithm}
For $x\in(0,1)$, define $a_x$ and $b_x$ to be the values of $a$ and $b$ obtained via a given method.
For example, for method II (mean at current value $x$ and concentration parameter $\alpha$), we have $a_x = \alpha x$ and $b_x = \alpha(1-x)$.
For method I (maximum density with target location $x$ and fixed variance $v$), $a_x$ and $b_x$ are the solution to \cref{eq:max_density} with $c = x$ and $h(a,b) = V(a,b)/v - 1$.

The Metropolis--Hastings algorithm then proceeds as follows. Let $\pi(x)$ denote the density of the target distribution, and suppose the current state is $x$.
Sample a proposed value $x'\sim\mathrm{Beta}(a_x,b_x)$.
Compute the acceptance probability
$$ p = \min\bigg\{1, \; \frac{\pi(x') \, \mathrm{Beta}(x\mid a_{x'},b_{x'})}{\pi(x) \, \mathrm{Beta}(x'\mid a_{x},b_{x})}\bigg\}. $$
With probability $p$, accept the proposal (so the state becomes $x'$), 
and otherwise, reject the proposal (so the state remains $x$).
We initialize the sampler by setting the initial state to be 0.25, and then each iteration of the MCMC algorithm consists of one MH move as described above.

% \subsection{Tuning the proposal distribution parameters}

\cref{fig:mcmc-tuning-alpha,fig:mcmc-tuning-max-density} show the results of tuning the parameters of each proposal distribution to optimize their performance in the MH example in \cref{sec:mh-example}. The plots show the autocorrelation functions (ACFs) for a range of values of $\alpha$ (for method II) and $v$ (for method I).

\begin{figure}
    \centering
    \includegraphics[width=\textwidth]{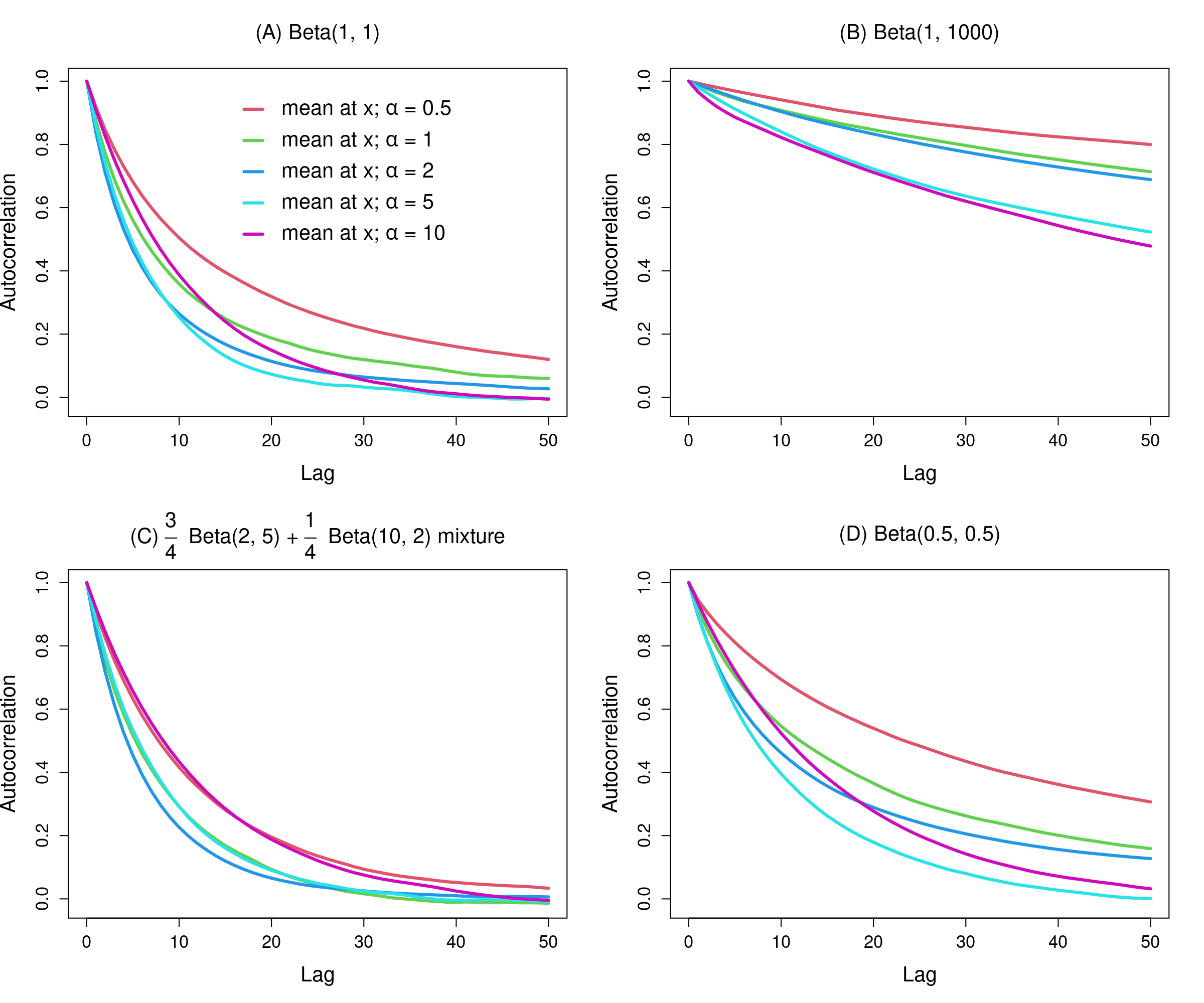}
    \caption{Tuning the concentration parameter of the mean method. Autocorrelation functions for Metropolis--Hasting samplers using proposal distribution II (mean equal to the current state $x$ and fixed concentration parameter $\alpha$) as in \cref{sec:mh-example}.}
    \label{fig:mcmc-tuning-alpha}
\end{figure}

\begin{figure}
    \centering
    \includegraphics[width=\textwidth]{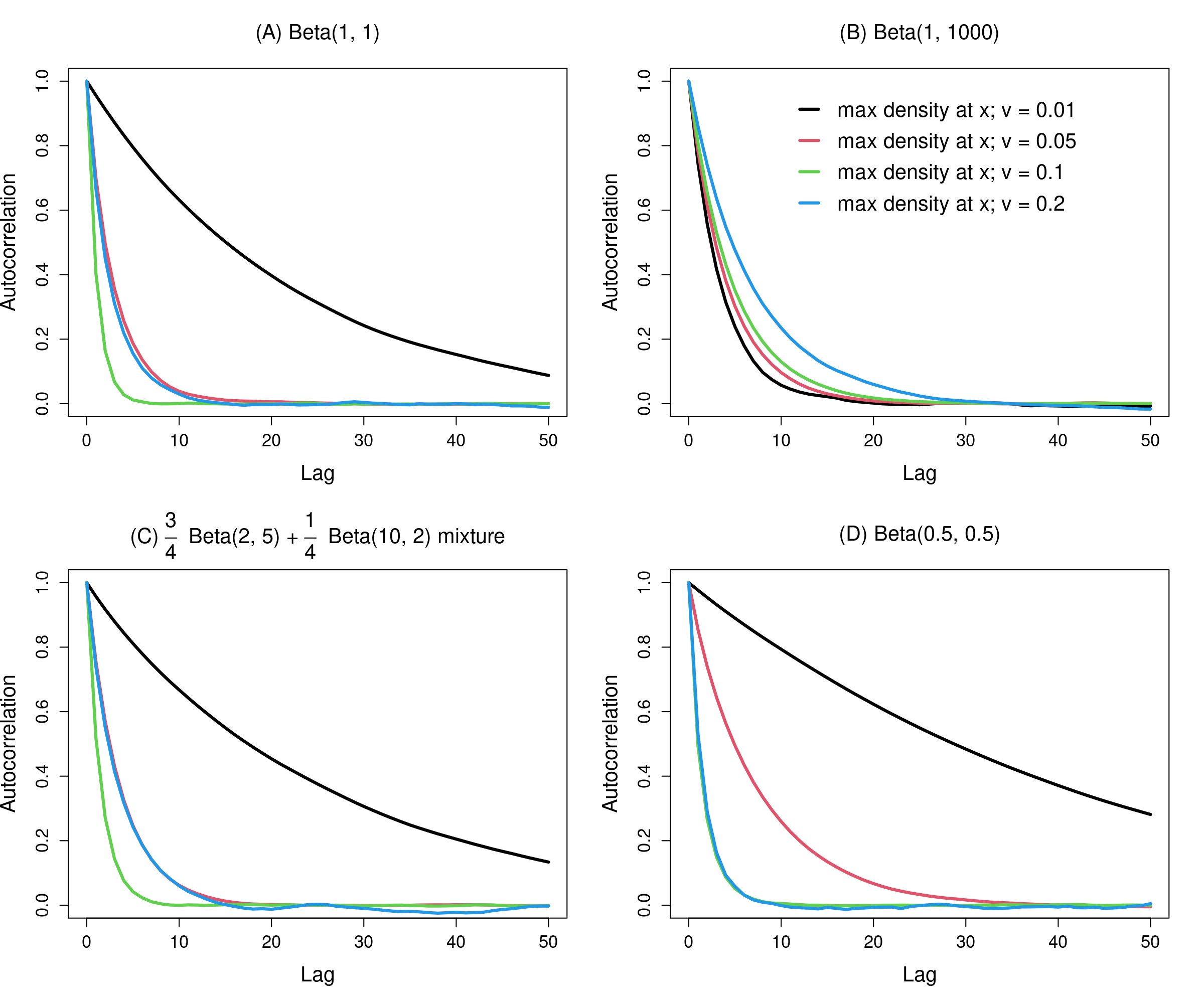}
    \caption{Tuning the variance parameter of the maximum density method. Autocorrelation functions for Metropolis--Hasting samplers using proposal distribution I (maximum density method with target location equal to the current state $x$ and fixed variance $v$) as in \cref{sec:mh-example}.}
    \label{fig:mcmc-tuning-max-density}
\end{figure}

\begin{figure}
    \centering
    \includegraphics[trim=0 0 0.25cm 2.25in, clip, width=0.7\textwidth]{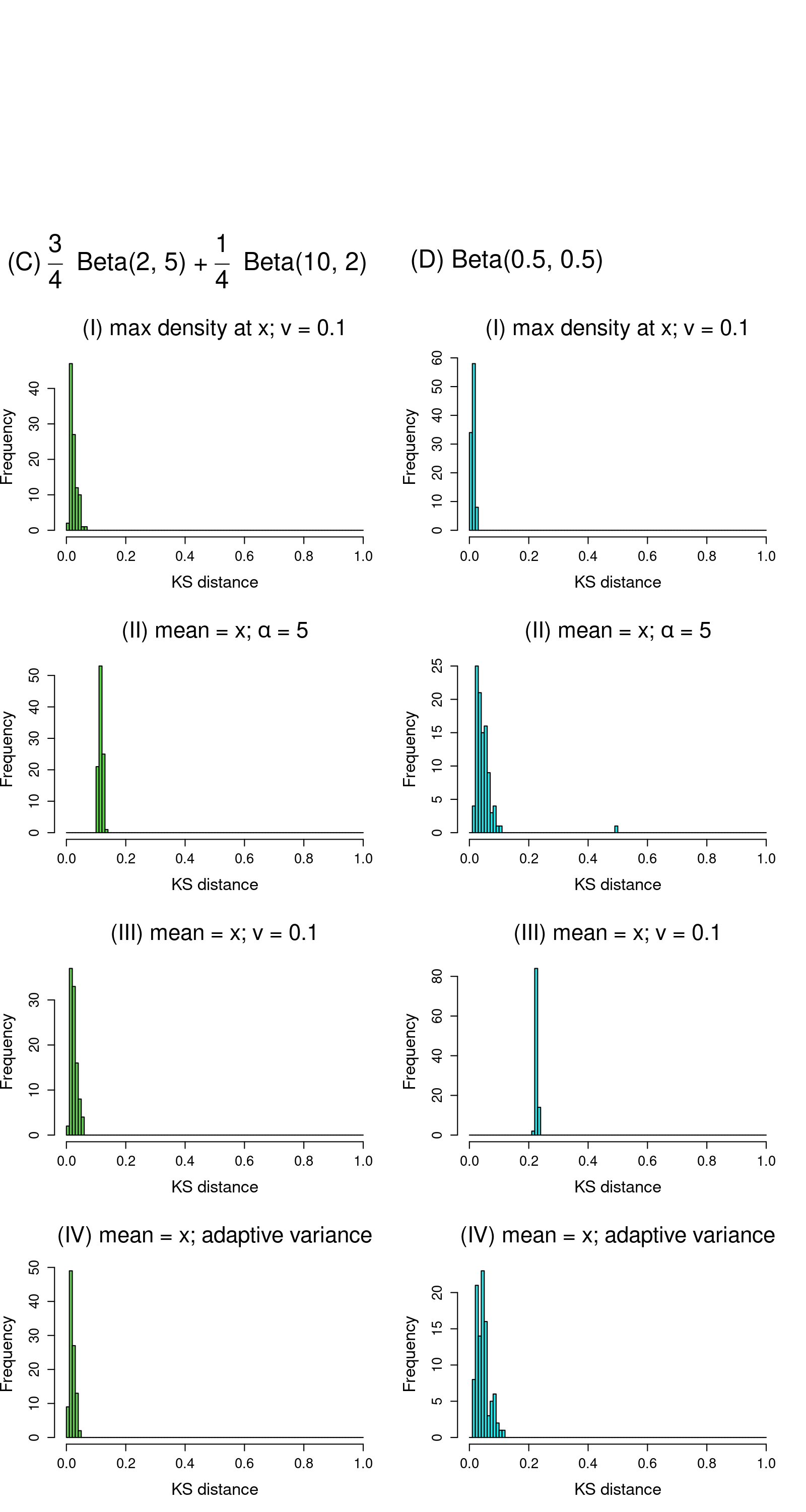}
    \caption{Histograms of the Kolmogorov--Smirnov distance between the target distribution and the MCMC approximation to it. 
    The green histograms in the left column are for the $0.75 \mathrm{Beta}(2,5) + 0.25 \mathrm{Beta}(10,2)$ target distribution (case C); the cyan histograms in the right column are for the $\mathrm{Beta}(0.5, 0.5)$ target distribution (case D).
    As in \cref{fig:ks-distance}, samplers using maximum density proposals (method I) or the mean method with adaptive variance (method IV) converge reliably to the target distribution,   
    while Methods II and III are less reliable depending on the specific target distribution.
    }
    \label{fig:ks-distance-suppl}
\end{figure}

\section{Taylor approximation to cosine error}
\label{sec:taylor-cosine}
The cosine error between points $x$ and $u$ in the probability simplex $\Delta_K$ is 
\begin{equation}
	\mathrm{CosErr}(x,u) = 1 - \frac{u^\mathtt{T} x}{\lVert u \rVert \lVert x \rVert} = 1 - \frac{\sum_i u_i x_i}{\sqrt{\sum_i u_i^2 \sum_i x_i^2}}. 
	\label{eq:cosine-error}
\end{equation}
To make this expression more tractable when taking its expectation over $x$, we employ a Taylor approximation to the cosine error as a function of $x$. To this end, fix $u\in\Delta_K$ and define $f(x) = \mathrm{CosErr}(x,u)$. Since $f$ is uniquely minimized at $x = u$ by the Cauchy--Schwarz inequality, we use the second-order Taylor approximation.  To derive this, we compute the first and second derivatives of $f$, which after simplifying are equal to
\begin{align}
	\frac{\partial f}{\partial x_i} &= - \frac{u_i}{\lVert u \rVert \lVert x \rVert} + \frac{(u^\mathtt{T} x) x_i}{\lVert u \rVert \lVert x \rVert^3} \label{eq:taylor-first-derivatives} \\
	% \frac{\partial^2 f}{\partial x_i^2} &= \frac{u_i x_i}{\lVert u \rVert \lVert x \rVert^3} + \frac{( x_i x_j + u^\mathtt{T} x ) \lVert u \rVert \lVert x \rVert^3 - u^\mathtt{T} x x_i \cdot \frac{3}{2} \lVert u \rVert \lVert x \rVert \cdot 2 x_i}{\lVert u \rVert^2 \lVert x \rVert^6} \\
 % &= \frac{2 u_i x_i + u^\mathtt{T} x}{\lVert u \rVert \lVert x \rVert^3} - \frac{3 u^\mathtt{T} x x_i^2}{\lVert u \rVert \lVert x \rVert^5} \\
	% \frac{\partial^2 f}{\partial x_i \, \partial x_j} &= \frac{u_i x_j}{\lVert u \rVert \lVert x \rVert^3} + \frac{u_j x_i \lVert u \rVert \lVert x \rVert^3 - u^\mathtt{T} x x_i \cdot \frac{3}{2} \lVert u \rVert \lVert x \rVert \cdot 2 x_j}{\lVert u \rVert^2 \lVert x \rVert^6} \\
\frac{\partial^2 f}{\partial x_i \, \partial x_j} &= \frac{u_i x_j + u_j x_i + (u^\mathtt{T} x)\mathds{1}(i=j)}{\lVert u \rVert \lVert x \rVert^3} - \frac{3 (u^\mathtt{T} x) x_i x_j}{\lVert u \rVert \lVert x \rVert^5}. \label{eq:taylor-second-derivatives}
\end{align}
Let $H = (H_{i j})\in\mathbb{R}^{K\times K}$ denote the Hessian matrix at $x = u$, that is,
\begin{align*}
    H_{i j} = \frac{\partial^2 f}{\partial x_i \, \partial x_j}(u) = \frac{\mathds{1}(i=j)}{\|u\|^2} - \frac{u_i u_j}{\|u\|^4}
\end{align*}
after plugging $x = u$ into \cref{eq:taylor-second-derivatives} and simplifying. Note that $H$ is symmetric and positive semi-definite since $f$ is minimized at $x=u$.
Also, $f(u) = 0$ and  $\partial f/\partial x_i = 0$ at $x = u$ by \cref{eq:cosine-error,eq:taylor-first-derivatives}. 
The second-order Taylor approximation around $u$ is therefore
\begin{align}
\label{eq:taylor-derivation}
	f(x) &\approx f(u) + \frac{1}{2} ( x - u )^\mathtt{T} H ( x - u ) \notag\\
 &= \frac{1}{2} \sum_{i = 1}^K \sum_{j = 1}^K ( x_i - u_i ) H_{i j} ( x_j - u_j ) \\
 &= \frac{1}{2} \sum_{i=1}^K (x_i - u_i)^2 \frac{1}{\|u\|^2} - \frac{1}{2} \sum_{i = 1}^K \sum_{j = 1}^K ( x_i - u_i ) ( x_j - u_j ) \frac{u_i u_j}{\|u\|^4}.  \notag
	% &= \sum_{i = 1}^K \sum_{j = 1}^{i - 1} ( x_i - u_i ) ( x_j - u_j ) \bigg( - \frac{u_i u_j}{\lVert u \rVert^4} \bigg) +  \frac{1}{2} \sum_{i = 1}^K ( x_i - u_i )^2 \bigg( - \frac{u_i^2}{\lVert u \rVert^4} + \frac{1}{\lVert u \rVert^2} \bigg).\notag
\end{align}

% Now, suppose $X = (X_1, \ldots, X_K) \sim \mathrm{Dirichlet}(a_1,\ldots,a_K)$, where $a_1,\ldots,a_K>0$.  Letting $\alpha = \sum_{i=1}^K a_i$ and $u_i = a_i/\alpha$, it follows that $\alpha$ is the concentration parameter and $u = \E(X)\in\Delta_K$ is the mean of $X$.

Now, suppose $X = (X_1, \ldots, X_K) \sim \mathrm{Dirichlet}(\alpha u_1,\ldots,\alpha u_K)$ where $\alpha > 0$.  Then $u = \E(X)$ is the mean of $X$, so taking the expectation of \cref{eq:taylor-derivation} over $X$ yields 
\begin{align}
\label{eq:taylor-var-cov}
    \E\big(f(X)\big) \approx \frac{1}{2\|u\|^2}\sum_{i=1}^K \mathrm{Var}(X_i) - \frac{1}{2\|u\|^4} \sum_{i=1}^K \sum_{j=1}^K u_i u_j \mathrm{Cov}(X_i,X_j).
\end{align}
Plugging in the formulas for the variances and covariances of the entries of a Dirichlet distributed random vector and simplifying, \cref{eq:taylor-var-cov} becomes
\begin{align}  
	\E\big(f(X)\big) \approx \frac{1}{2 (1+\alpha) \lVert u \rVert^2} \bigg( 1 - \frac{\sum_i u_i^3}{\sum_i u_i^2} \bigg). \label{eq:taylor-approx}
\end{align}
Finally, letting $a_i = \alpha u_i$ for $i=1,\ldots,K$, \cref{eq:taylor-approx} can be written as
\begin{align}
\begin{split}
\E(\mathrm{CosErr}(X,\E(X))) = \E(f(X)) &\approx \frac{1}{2 (1 + \alpha) \lVert a/\alpha \rVert^2} \bigg( 1 - \frac{\sum_i (a_i/\alpha)^3}{\sum_i (a_i/\alpha)^2} \bigg) \\
&=  \frac{\sum_i a_i}{2 (1 + \sum_i a_i) (\sum_i a_i^2)} \bigg({\textstyle\sum_i} a_i - \frac{\sum_i a_i^3}{\sum_i a_i^2} \bigg).
\end{split}
\end{align}

\bibliographystyle{ba}
\bibliography{references}

\end{document}